\documentclass[11pt,intlimits,reqno]{amsart}

\usepackage{amsmath}

\usepackage{enumerate}

\usepackage{amssymb}

\usepackage{amsthm}

\usepackage{bbm}

\usepackage{graphicx}

\usepackage{mathrsfs}
\usepackage{tikz}
\usetikzlibrary{matrix}
\usepackage{tikz-cd}
\usetikzlibrary{cd}

\usepackage[polish,english]{babel}
\usepackage{verbatim}
\usepackage[utf8]{inputenc}
\usepackage{bbm}

\newtheorem{thm}{Theorem}[section]

\newtheorem{prop}[thm]{Proposition}

\theoremstyle{definition}

\newtheorem{defn}[thm]{Definition}

\theoremstyle{remark}

\numberwithin{equation}{section}

\begin{document}

\title{Integrable Hamiltonian systems on the symplectic realizations of $\textbf{e}(3)^*$}

\maketitle 

\begin{center}
A. Odzijewicz, E. Wawreniuk\\
Department of Mathematics\\
University in Bia\l{}ystok\\
Cio\l{}kowskiego 1M, 15-245 Bia\l{}ystok, Poland\\
aodzijew@uwb.edu.pl , ewawreniuk@math.uwb.edu.pl
\end{center}

\begin{abstract}
The phase space of a gyrostat with a fixed point and a heavy top is the Lie-Poisson space $\textbf{e}(3)^*\cong \mathbb{R}^3\times \mathbb{R}^3$ dual to the Lie algebra $\textbf{e}(3)$ of Euclidean group $E(3)$. One has three naturally distinguished Poisson submanifolds of $\textbf{e}(3)^*$: (i) the dense open submanifold $\mathbb{R}^3\times \dot{\mathbb{R}}^3\subset \textbf{e}(3)^*$ which consists of all $4$-dimensional symplectic leaves ($\vec{\Gamma}^2>0$); (ii) the $5$-dimensional Poisson submanifold of $\mathbb{R}^3\times \dot{\mathbb{R}}^3$ defined by $\vec{J}\cdot \vec{\Gamma} = \mu ||\vec{\Gamma}||$; (iii) the $5$-dimensional Poisson submanifold of $\mathbb{R}^3\times \dot{\mathbb{R}}^3$ defined by $\vec{\Gamma}^2 = \nu^2$, where $\dot{\mathbb{R}}^3:= \mathbb{R}^3\backslash \{0\}$, $(\vec{J}, \vec{\Gamma})\in \mathbb{R}^3\times \mathbb{R}^3\cong \textbf{e}(3)^*$ and $\nu < 0 $, $\mu$ are some fixed real parameters. Basing on the $U(2,2)$-invariant symplectic structure of Penrose twistor space  we find  full and complete $E(3)$-equivariant symplectic realizations of these Poisson submanifolds which are $8$-dimensional for (i) and $6$-dimensional for (ii) and (iii). As a consequence of the above Hamiltonian systems on $\textbf{e}(3)^*$ lift to the ones on the above symplectic realizations. In such a way after lifting integrable cases of gyrostat with a fixed point, as well as of heavy top, we obtain a large family of integrable Hamiltonian systems on the phase spaces defined by these symplectic realizations. 
 \end{abstract}

\tableofcontents

\section{Introduction}

The group $U(2,2)$, which contains Poincare and Euclidean groups as subgroups, occupies an exceptional place in mathematical physics. First of all it is the symmetry group of conformal compactification $\bar{M}_{1,3} \cong U(2)$ of Minkowski space-time $M_{1,3}$ \cite{P}.  Hence, it is the underlying group of Penrose twistor theory \cite{P}, where one  considers the space-time points as $2$-dimensional isotropic subspace of the twistor space $\mathbb{T}$. On the other hand, the complexification $\bar{M}_{1,3}^{ \mathbb{C}}$ and the cotangent bundle $T^*\bar{M}_{1,3}$ of $\bar{M}_{1,3}$ could be considered as the phase spaces of conformal particles \cite{O2}. The above leads to investigation of a model of field theory on $\bar{M}_{1,3}^{ \mathbb{C}}$alternative to the one on the real Minkowski space $M_{1,3}$ \cite{O}. Let us mention also that $U(2,2)$ is the dynamical group for the regularized Kepler and MIC-Kepler systems \cite{I,OSW,OS}.

In this paper we will show that $U(2,2)$ plays an important role in the description of the dynamics of a gyrostat with a fixed point and heavy top dynamics in particular. Namely, taking appropriate decomposition of the Lie-Poisson space $\textbf{u}(2,2)^*$ on the subspaces including the dual space $\textbf{e}(3)^*$ of the Lie algebra $\textbf{e}(3)$ of Euclidean group $E(3)$, which is underlying structure of the rigid body theory, we construct four symplectic realizations of the Lie-Poisson space $\textbf{e}(3)^*$ in sense of \cite{W,ZU}. The above allow us to find various integrable Hamiltonian systems on these  symplectic realizations corresponding to the integrable cases of the gyrostat and heavy top dynamics on $\textbf{e}(3)^*$, see for example \cite{BF} and references therein.

The content of the paper is the following.

For self-sufficiency of the paper in Section \ref{sec:1} we present some preliminaries of the theory of  gyrostat and rigid body spinning around of a fixed point.

In Section \ref{sec:2} we discuss necessary questions concerning the twistor space $\mathbb{T}$ as $U(2,2)$-symplectic space. 

In Section \ref{sec:3} using the anti-diagonal spinor representation of the twistor space $\mathbb{T}$ we decompose, see \eqref{419}, the Lie algebra  $ \textbf{u}(2,2)\cong\textbf{u}(2,2)^*$ on the subspaces dual to appropriately chosen Lie subalgebras what allow us to define the respective momentum maps $\textbf{J}_u:\mathbb{T}\to (\textbf{e}(3)\oplus \textbf{a}(2))^*$, $\textbf{J}_e: \mathbb{T}\to \textbf{e}(3)^*$ and $\textbf{J}_a: \mathbb{T}\to \textbf{a}(2)^*$ as superpositions of the conformal group momentum map $\textbf{J}:\mathbb{T} \to \textbf{u}(2,2)^*\cong \textbf{u}(2,2)$ with projections on the components of the above decomposition. Thereby, we show that $\textbf{J}_e: \widetilde{\mathbb{T}}:= \mathbb{T}\backslash \textbf{J}^{-1}_a(0) \to \mathbb{R}^3\times\dot{\mathbb{R}}^3\subset \mathbb{R}^3\times \mathbb{R}^3 \cong \textbf{e}(3)^*$ is a full and complete symplectic realization of the open dense Poisson submanifold $\mathbb{R}^3\times \dot{\mathbb{R}}^3$ of $\textbf{e}(3)^*$ consisting of $4$-dimensional symplectic leaves of $\textbf{e}(3)^*$, see Proposition \ref{prop:41}, Proposition \ref{prop:2} and Corollary \ref{cor:44}. In the diagram \eqref{diagram62} of Proposition \ref{prop:2} the dual symplectic pair involving Poisson submanifold $\mathbb{R}^3\times \dot{\mathbb{R}}^3\subset \textbf{e}(3)^*$ is also presented.

In Section \ref{sec:4} we apply symplectic reduction procedure to the levels of the maps $J_0:\widetilde{\mathbb{T}}\to \mathbb{R}$, $
\Gamma_0:\widetilde{\mathbb{T}}\to \mathbb{R}$ and $\textbf{J}_a:\widetilde{\mathbb{T}}\to \textbf{a}(2)^*$, defined in \eqref{428}, \eqref{429} and \eqref{435}, respectively. As a result  we obtain three other $\textbf{J}_{e, \mu}: \mathbb{R}^3\times \dot{\mathbb{R}}^3\to \textbf{e}(3)^*$, $ \textbf{J}_{e, \nu}: T^*\mathbb{S}^3_{\rho}\to \textbf{e}(3)^*$ and $\textbf{J}_{e, \mu , \nu}: M_{\mu, \nu }\to \textbf{e}(3)^*$ $E(3)$-equivariant  symplectic realizations of $\textbf{e}(3)^*$. See also diagram \eqref{538}.  The Poisson geometric structure of these realizations is investigated too.

In Section \ref{sec:5}  we present a family of integrable Hamiltonian systems on symplectic manifolds $(\widetilde{\mathbb{T}}, \Omega)$, $(\mathbb{R}^3\times \dot{\mathbb{R}}^3, \Omega_\mu )$ and $(T^*\mathbb{S}^3_\rho, \Omega_\nu)$, being the liftings of the integrable cases of gyrostat and heavy top on the above symplectic manifolds, e.g. see \cite{BF}. Some physical interpretation of these Hamiltonian systems is discussed too.

For self-sufficiency of the paper we include short appendices $A$ and $B$ containing some notions and facts from Poisson Geometry.

\section{Preliminaries on heavy top}\label{sec:1}

The Lie algebra $\textbf{e}(3)$ of the Euclidean group $E(3)$ underlies the rigid body theory, i.e. the description of motion of a gyrostat with a fixed point (heavy top in particular) in gravitational field, see e.g. \cite{BF,RM} and references therein. Omitting details, one can consider the gyrostat with a fixed point as a Hamiltonian system on the Lie-Poisson space $\textbf{e}(3)^* \cong \mathbb{R}^3 \times \mathbb{R}^3$ with Hamiltonian in general given by 
\begin{multline}\label{Hht}
H_{\lambda} = \frac{(J_1+ \lambda_1)^2}{2I_1}+ \frac{(J_2+ \lambda_2)^2}{2I_2}+\frac{(J_3+ \lambda_3)^2}{2I_3}+ U(\vec{\Gamma})= \\
 =\frac{1}{2} (\vec{J}+\vec{\lambda})^TI^{-1}  (\vec{J}+\vec{\lambda }) + U(\vec{\Gamma}),
\end{multline}
where $\vec{\lambda}: \mathbb{R}^3 \to \mathbb{R}^3$ and $U:\mathbb{R}^3 \to \mathbb{R}$ are arbitrary smooth functions of  $\vec{\Gamma}\in \mathbb{R}^3$. The fixed parameters $I_1, I_2$ and $I_3$ are the principal moments of inertia and
\begin{equation}
I:= \left(\begin{array}{ccc}
I_1 & 0 & 0 \\
0 & I_2 & 0\\
0 & 0 & I_3
\end{array}\right).
\end{equation}
In the paper we will consider only the case $\vec{\lambda }(\Gamma ) =const$. 
The Lie-Poisson brackets of the dynamical coordinates functions $(\vec{J}, \vec{\Gamma})\in \mathbb{R}^3 \times \mathbb{R}^3 \cong (\textbf{e}(3)^*)^*\cong \textbf{e}(3)$ on the dual space $\textbf{e}(3)^*$ of $\textbf{e}(3)$ are defined by 
\begin{equation}\label{LPBjg}
\{J_k, J_l\}_{_{LP}} = \epsilon_{klm}J_m, \quad \{J_k, \Gamma_l \}_{_{LP}}= \epsilon_{klm}\Gamma_m, \quad \{\Gamma_k, \Gamma_l \}_{_{LP}} =0, 
\end{equation}
where $k,l=1,2,3$.  The  vector quantities $\vec{J}$ and $\vec{\Gamma }$ are the body angular momentum with respect to the fixed point and the gravity field, respectively, as seen from the body. If $U(\vec{\Gamma}) = \vec{\chi}\cdot \vec{\Gamma}$ then $\vec{\chi}$  is the fixed  body  vector on the line segment connecting the body fixed point with its center of a mass multiplied by gravitational acceleration $g$ and total mass $M$ of the body. In the case when the rigid body has cavities entirely filled by an incompressible perfect fluid then its angular momentum is $\vec{J} + \vec{\lambda}$, where the constant vector $\vec{\lambda}$ describes the cyclic motion of the fluid in cavities. The formula on angular momentum remains the same if instead of the liquid in body cavities one consider a gyrostat, i.e. a flywheel in the body fixed in such a way that $\vec{\lambda}$ is the direction of its axis and the body mass density does not depend on the wheel position \cite{Z}. For the case $\vec{\lambda}=0$ and potential $U$ as a quadratic function of $\vec{\Gamma}$ one obtains the dynamics of a rigid body in a liquid \cite{C}.

From \eqref{LPBjg} one easily sees that
\begin{equation}\label{casimirs}
K_1(\vec{J}, \vec{\Gamma}):=\vec{\Gamma }\cdot \vec{J} \qquad\mbox{ and }\qquad K_2(\vec{J}, \vec{\Gamma}):=\vec{\Gamma}^2 
\end{equation}
are Casimir functions for the Lie-Poisson algebra  $(C^\infty(\textbf{e}(3)^*, \mathbb{R}), \{\cdot , \cdot \}_{_{LP}})$ defined by (\ref{LPBjg}), i.e. one has $\{K_1, F\}_{_{LP}}=0= \{K_2 , F\}_{_{LP}}$ for arbitrary function $F= F(\vec{J}, \vec{\Gamma})$ of the dynamical variables $(\vec{J}, \vec{\Gamma})\in \mathbb{R}^3 \times \mathbb{R}^3$. So, for integration of the system one needs only one integral of motion $K=K(\vec{J}, \vec{\Gamma})$ except of the Hamiltonian (\ref{Hht}). For the most general case, when Hamiltonian $H=H(\vec{J}, \vec{\Gamma})$ is an arbitrary function of $\vec{J}$ and $\vec{\Gamma}$,  the heavy top dynamics is described by the Hamilton equations 
\begin{equation}\label{hameq1}
\begin{array}{l}
\frac{d}{dt}\vec{J} = \vec{J}\times \frac{\partial H}{\partial \vec{J}} + \vec{\Gamma}\times \frac{\partial H}{\partial \vec{\Gamma}}, \\
\frac{d}{dt}\vec{\Gamma} = \vec{\Gamma} \times \frac{\partial H}{\partial \vec{J}}.
\end{array}
\end{equation}
Taking Hamiltonian \eqref{Hht} in \eqref{hameq1} we obtain Hamilton equations 
\begin{equation}
\begin{array}{l}
\frac{d}{dt}\vec{J} = \vec{J}\times I^{-1}(\vec{J}+\vec{\lambda}) + \vec{\Gamma}\times \frac{\partial U}{\partial \vec{\Gamma}}, \\
\frac{d}{dt}\vec{\Gamma} = \vec{\Gamma} \times I^{-1}(\vec{J}+\vec{\lambda}).
\end{array}
\end{equation}
for a gyrostat with a fixed point.

There are known many famous integrable cases of the rigid body systems. Following of \cite{BF} we mention some of them in the chronological order as they were investigated by: L. Euler (1750), J.L. Lagrange (1788),  A. Clebsch (1871), N. Zhukovskii (1885),  V. Steklov and A. Lyapunov (1893), S. Kovalevskaya  (1899), D. Goryachev and S.Chaplygin (1899),  L.N. Sretenski (1963), S. Kovalevskaya and H.M. Yahia (1986).

For instance:\\
(i) One obtains the Kovalevskaya top \cite{SK} taking in \eqref{Hht} $\vec{\lambda} = 0$ and putting two moments of inertia equal and twice as large as the third one $I=I_1 = I_2 = 2I_3$. In this case the center of mass lies in the equatorial plane related to the coinciding axes of the inertia ellipsoid
\begin{equation}\label{Kht}
H_K = \frac{J_1^2}{2I} +\frac{J_2^2}{2I} + \frac{J_3^2}{I} + \chi_1 \Gamma_1+ \chi_2\Gamma_2.
\end{equation}
The fourth integral of motion is the famous Kovalevskaya invariant 
\begin{equation}
K(\vec{J}, \vec{\Gamma}) = \left(\frac{J_1^2 - J_2^2 }{2I} + \chi_2 \Gamma_2 - \chi_1 \Gamma_1\right)^2 +\left(\frac{J_1J_2}{I}- \chi_1 \Gamma_2 - \chi_2 \Gamma_1\right)^2.
\end{equation} 
(ii) One obtains the Zhukovskii case \cite{Z} taking in \eqref{Hht} $U(\vec{\Gamma}) =0$, i.e. 
\begin{equation}\label{HZ}
H_Z = \frac{(J_1+ \lambda_1)^2}{2I_1}+ \frac{(J_2+ \lambda_2)^2}{2I_2}+\frac{(J_3+ \lambda_3)^2}{2I_3} . 
\end{equation}
It describes the motion of gyrostat in a gravity field when its fixed point is in the center of a mass. The additional integral of motion is given by 
\begin{equation}
K(\vec{J}, \vec{\Gamma}) = \vec{J}^2. 
\end{equation}
Taking in \eqref{HZ} $\vec{\lambda}=0$ one obtains the Euler top, i.e. a rigid body with fixed point in its center of mass. \\
(iii) One obtains the Clebsh case \cite{C} taking in \eqref{Hht} $\vec{\lambda}=0$ and potential $U(\vec{\Gamma})$ as a quadratic function of $\Gamma_1, \Gamma_2, \Gamma_3$, so that
\begin{equation}
H_C = \frac{J_1^2}{2I_1} +\frac{J_2^2}{2I_2} + \frac{J_3^2}{2I_3} + \frac{\epsilon}{2}(I_1 \Gamma_1^2 + I_2\Gamma_2^2 + I_3 \Gamma_3^2), 
\end{equation}
for $\epsilon \in \mathbb{R}$. The above Hamiltonian describes the motion of a rigid body in a fluid. The quadratic polynomial
\begin{equation}
K(\vec{J}, \vec{\Gamma}) = \frac{1}{2} \vec{J}^2 - \frac{\epsilon}{2}( I_2I_3 \Gamma_1^2 + I_3I_1 \Gamma_2^2 + I_1I_2 \Gamma_3^2)
\end{equation}
is the additional integral of motion for this case.

At the end, let us also stress that the classical version of the Lipkin-Meshkov-Glick Hamiltonian
\begin{equation}\label{Hl}
H_{LMG} = \varepsilon J_3 + \frac{1}{2}V (J_+^2 + J_-^2) + \frac{1}{2}W (J_+J_- + J_-J_+),
\end{equation}
can be obtained from the Hamiltonian (\ref{HZ})  if one takes $\frac{1}{I_1}= W+V$, $\frac{1}{I_2} = W-V$, $\vec{\lambda} = (0,0,\lambda_3)^T$ and the limits $\frac{1}{I_3} \to 0$, $\frac{\lambda_3}{I_3}\to \varepsilon = const$, where the constant $\frac{\lambda_3^2}{2I_3}$ was neglected. 
Quantum version of the above Hamiltonian is  known in nuclear physics as Lipkin-Meshkov-Glick model of nuclei \cite{L}. 

\section{Penrose twistor space as a symplectic $U(2,2)$-manifold}\label{sec:2}

Following \cite{OSW}, in this section we describe the $U(2,2)$-invariant symplectic structure on Penrose twistor space \cite{P} denoted here by $\mathbb{T}$   and the respective momentum map of $\mathbb{T}$ into the dual $\textbf{u}(2,2)^*$ of Lie algbera $\textbf{u}(2,2)$ of the group $U(2,2)$.  The explicit forms of Hamilton equation on $\mathbb{T}$ and on $\textbf{u}(2,2)^*$ are presented too.

Let us recall that twistor space $\mathbb{T}$ by definition is $\mathbb{C}^4$   equipped with the Hermitian scalar product 
\begin{equation}\label{sp1}
\langle v, w \rangle := v^+\phi w =  \bar{v}^k \phi_{kl}w^l
\end{equation}
of $v,w \in \mathbb{C}^4$, $k,l=1,2,3,4$, where $\phi=[\phi_{kl}] \in \mbox{Mat}_{4\times 4 } (\mathbb{C})$ is a  Hermitian matrix $\phi^+ = \phi $ of signature $(++--)$. Additionally we will assume that $\phi^2 = \mathbbm{1}_4$.  In \eqref{sp1}, as well as in the further part of this paper, we use the summation convention. The symmetry group of $\mathbb{T} \cong (\mathbb{C}^4, \phi )$ is the group $U(2,2) \subset GL(4, \mathbb{C})$, i.e. $g \in U(2,2)$ if and only if 
\begin{equation}
\langle gv, gw \rangle = \langle v, w\rangle 
\end{equation}
for any $v, w \in \mathbb{C}^4$. In \eqref{sp1} by "$+$" we denoted Hermitian conjugation map $^+: \mathbb{C}^4 \ni v \mapsto v^+:= \bar{v}^T \in (\mathbb{C}^4)^*$. We will use also the twistor conjugation map
\begin{equation}\label{cong}
\mathbb{T} \ni w \mapsto w^* := w^+\phi \in \mathbb{T}^*, \quad w^*_l =  \bar{w}^k \phi_{kl}
\end{equation}
consistent with the scalar product (\ref{sp1}) by $\langle v, w \rangle = v^*w $. The twistor conjugation defines the conjugation of elements of $\mbox{End}(\mathbb{T})\cong \mathbb{T}\otimes \mathbb{T}^*$ by 
\begin{equation}
\mbox{End} (\mathbb{T}) \ni A \mapsto  A^* := \phi A^+ \phi \in \mbox{End}(\mathbb{T}).
\end{equation}
Hence, we see that
\begin{equation}
g \in U(2,2) \mbox{ iff } g^{-1} = g^* \quad \mbox{ and } \quad \mathcal{X} \in \textbf{u}(2,2) \mbox{ iff } \mathcal{X} + \mathcal{X}^* = 0.
\end{equation}
One has the non-singular $\mbox{Ad}_{U(2,2)}$-invariant scalar product 
\begin{equation}
\textbf{u}(2,2)\times \textbf{u}(2,2) \ni (\mathcal{X}, \mathcal{Y}) \mapsto -\mbox{Tr}( \mathcal{X}^* \mathcal{Y})= \mbox{Tr}(\mathcal{X}\mathcal{Y}) \in \mathbb{R}
\end{equation}
on the real Lie algebra $\textbf{u}(2,2)$ which defines $U(2,2)$-equivariant isomorphism 
\begin{equation}\label{linfun}
\mathfrak{I} : \textbf{u}(2,2) \ni \mathcal{X} \mapsto \mathfrak{I}(\mathcal{X}) =  \mbox{Tr}(\mathcal{X} \cdot ) \in \textbf{u}(2,2)^*
\end{equation}
between the $\textbf{u}(2,2)$ and its dual $\textbf{u}(2,2)^*$, i.e. 
\begin{equation}
\mathfrak{I}\circ \mbox{Ad}_g = \mbox{Ad}^*_{g^{-1}} \circ \mathfrak{I} \mbox{ for } g\in U(2,2).
\end{equation} 
So, in subsequent we will identify $\textbf{u}(2,2)^* \cong \textbf{u}(2,2)$ the dual space $\textbf{u}(2,2)^*$ with $\textbf{u}(2,2)$. Let us note that $\overline{\mbox{Tr}(\mathcal{X}^*Y)}= \mbox{Tr}(\mathcal{Y}^+ \phi \mathcal{X}\phi)=\mbox{Tr}(\mathcal{Y}^*\mathcal{X}) =  \mbox{Tr}(-\mathcal{X}\mathcal{Y}) = \mbox{Tr}(\mathcal{X}^*Y)$ for $\mathcal{X}, \mathcal{Y} \in \textbf{u}(2,2)$.

Taking the above into account, one can consider $\textbf{u}(2,2)$ as the Lie-Poisson space with the Lie-Poisson bracket of $F,G \in C^\infty (\textbf{u}(2,2), \mathbb{R})$ defined by 
\begin{equation}\label{LPbracket}
\{F, G \}_{_{LP}} (\rho) := \mbox{Tr}\left(\rho \left[\frac{\partial F}{\partial \rho}(\rho), \frac{\partial G}{\partial \rho}(\rho)\right]\right)
\end{equation}
where $\rho \in \textbf{u}(2,2)$ and $\frac{\partial F}{\partial \rho}(\rho), \frac{\partial G}{\partial \rho}(\rho) \in \textbf{u}(2,2) $ are derivatives  of $F,G$ taken at $\rho$. Recall here that $\rho^* = - \rho$.  We will denote by $\rho^k_l$, where $k,l=1,2,3,4$, the matrix elements of $\rho=[\rho^k_l]\in\textbf{u}(2,2) \subset \mathbb{T}\otimes \mathbb{T}^* \cong \mbox{End}(\mathbb{T})$. The formula for Lie-Poisson bracket (\ref{LPbracket})  in these coordinates assumes the form
\begin{equation}
\{F, G \}_{_{LP}} (\rho)= \rho^k_l \left(\frac{\partial F}{\partial \rho^n_l}\frac{\partial G}{\partial \rho^k_n} - \frac{\partial G}{\partial \rho^n_l}\frac{\partial F}{\partial \rho^k_n}\right) . 
\end{equation}

For linear functions $\mathfrak{I}(\mathcal{X}), \mathfrak{I}(\mathcal{Y})\in \textbf{u}(2,2)^*\subset C^\infty (\textbf{u}(2,2), \mathbb{R}) $ one has $\frac{\partial\mathfrak{I}(\mathcal{X})}{\partial \rho}(\rho) = \mathcal{X}, \frac{\partial\mathfrak{I}(\mathcal{Y})}{\partial \rho}(\rho) = \mathcal{Y}$.  Therefore, one easily sees that 
\begin{equation}\label{lpnawias}
\{\mathfrak{I}(\mathcal{X}), \mathfrak{I}(\mathcal{Y})\}_{_{LP}} (\rho) = \mathfrak{I}([\mathcal{X}, \mathcal{Y}])(\rho). 
\end{equation}

Let us mention that the Hamilton equation on the Lie-Poisson space $(\textbf{u}(2,2), \{\cdot , \cdot \}_{_{LP}})$ corresponding to a Hamilton function $H\in C^\infty (\textbf{u}(2,2), \mathbb{R})$ is the following
\begin{equation}\label{hameq:u22}
\frac{d}{dt} \rho = [\rho , \frac{\partial H}{\partial \rho}(\rho)]. 
\end{equation}

The $U(2,2)$-invariant symplectic form on $\mathbb{T} = (\mathbb{C}^4, \phi)$ one defines by 
\begin{equation}\label{1form}
\Omega:=  -id( w^*dw ).= -i dw^* \wedge dw= -i dw^*_k \wedge dw^k.
\end{equation}
Hence, the corresponding Poisson bracket of $f,g \in C^\infty (\mathbb{T}, \mathbb{R})$ and the momentum map $\textbf{J}: \mathbb{T} \to \textbf{u}(2,2)$ are given by 
\begin{equation}\label{pbf}
\{f, g \} (w, w^*) = -i \left[\frac{\partial f}{\partial w}\frac{\partial g}{\partial w^*} - \frac{\partial g}{\partial w}\frac{\partial f}{\partial w^*}\right]= -i \left[\frac{\partial f}{\partial w^k}\frac{\partial g}{\partial w^*_k} - \frac{\partial g}{\partial w^k}\frac{\partial f}{\partial w^*_k}\right]
\end{equation}
and by 
\begin{equation}\label{mm1}
\textbf{J}(w, w^+) = i ww^*, \quad \textbf{J}^k_l = iw^kw^*_l,
\end{equation}
respectively. Let us mention also that $\textbf{J}: \mathbb{T} \to \textbf{u}(2,2)$ is a $U(2,2)$-equivariant 
\begin{equation}
\textbf{J}(gw, (gw)^*) = g \textbf{J}(w, w^*) g^*
\end{equation}
Poisson map 
\begin{equation}
\{F\circ \textbf{J}, G\circ \textbf{J} \} = \{F, G\}_{_{LP}} \circ \textbf{J},
\end{equation}
which satisfies 
\begin{equation}
\textbf{J}(w, w^*)^2 = iw^*w \textbf{J}(w, w^*)
\end{equation}
and 
\begin{equation}\label{316}
(\mathfrak{I}(\mathcal{X}) \circ \textbf{J})(w, w^* ) = i w^* \mathcal{X} w. 
\end{equation}
If one takes the Hamiltonian $h = H\circ \textbf{J}$, then the flow $\sigma_t^h$ on the twistor space $\mathbb{T}$ defined by the Hamilton equations 
\begin{equation}\label{eq:319}
\begin{array}{ll}
\frac{d}{dt} w = \{h, w\}= i \frac{\partial h}{\partial w^*}, \quad & \frac{d}{dt} w^k = \{h, w^k\}= i \frac{\partial h}{\partial w^*_k}\\
\frac{d}{dt} w^* = \{h, w^*\}= -i \frac{\partial h}{\partial w}, \quad & \frac{d}{dt} w^*_k = \{h, w^*_k\}= i \frac{\partial h}{\partial w^k},
\end{array}
\end{equation}
for $h \in C^\infty (\mathbb{T} , \mathbb{R})$ and the flow $\sigma_t^H$ on $\textbf{u}(2,2)$ defined by \eqref{hameq:u22} for $H$ are $\textbf{J}$-related, i.e. $\textbf{J}\circ \sigma_t^h = \sigma_t^H \circ \textbf{J}$. So, one could consider \eqref{hameq:u22} as a Lax representation of \eqref{eq:319}.

\section{Twistor space as a symplectic realization of heavy top phase space}\label{sec:3}

We recall, see e.g. \cite{Z,W}, that  a symplectic realization of a Poisson space $(P, \{\cdot , \cdot \})$  is by definition a Poisson map $\Phi: M \to P$ from a symplectic manifold $(M, \omega)$. If $\Phi: M \to P$ is a surjective submersion then one calls it a full symplectic realization of $P$. In this section we will show that twistor space $(\mathbb{T}, \Omega)$ is a symplectic realization of the phase space $(\textbf{e}(3)^*, \{\cdot , \cdot \}_{_{LP}})$ of a gyrostat with a fixed point (heavy top in particular case).

For this purpose
we will use 
the anti-diagonal $\mathbb{T} = (\mathbb{C}^2 \oplus \mathbb{C}^2 , \phi)$ spinor representation of the twistor space $\mathbb{T} \cong (\mathbb{C}^4, \phi)$, which  is defined in block matrix form as follows 
\begin{equation}\label{eq:41}
\phi = i\left(\begin{array}{cc}
\textbf{0}  & -\sigma_0 \\
\sigma_0 & \textbf{0} \end{array}\right) \mbox{ and } w =\left(\begin{array}{c}
\vartheta \\
\zeta
\end{array}\right), 
\end{equation}
where $\sigma_0$,  $\textbf{0}$ are unit and zero $2\times 2$ matrices and $\vartheta , \zeta \in \mathbb{C}^2$ are spinor coordinates of the twistor $w\in \mathbb{C}^2 \oplus \mathbb{C}^2$.
We will also use the Pauli matrices 
\begin{equation}\label{pauli}
\sigma_0 = \left(\begin{array}{cc}
1 & 0 \\
0 & 1 \end{array}\right) , \quad \sigma_1 = \left(\begin{array}{cc}
0 & 1 \\
1 & 0 \end{array}\right), \quad \sigma_2= \left(\begin{array}{cc}
0 & i \\
-i & 0 \end{array}\right), \quad \sigma_3= \left(\begin{array}{cc}
1 & 0 \\
0 &-1  \end{array}\right)
\end{equation}
as a basis of the real vector space $H(2)$ of the Hermitian $2\times 2$-matrices. 

In anti-diagonal representation of the twistor space $\mathbb{T}$ the symplectic form (\ref{1form}), the Poisson bracket (\ref{pbf}) and momentum map (\ref{mm1})  are given by 
\begin{equation}\label{2form}
\Omega:= d\Theta= d\zeta^+ \wedge d\vartheta - d\vartheta^+\wedge d\zeta , 
\end{equation}
where 
\begin{equation}\label{thetaa}
\Theta := \zeta^+d\vartheta - \vartheta^+d\zeta, 
\end{equation}
\begin{equation}\label{antipb}
\{f, g \} (\vartheta, \zeta, \vartheta^+, \zeta^+) = \frac{\partial f}{\partial \zeta^+}\frac{\partial g}{\partial \vartheta}- \frac{\partial g}{\partial \zeta^+}\frac{\partial f}{\partial \vartheta} - \left(\frac{\partial f}{\partial \vartheta^+}\frac{\partial g}{\partial \zeta}-\frac{\partial g}{\partial \vartheta^+}\frac{\partial f}{\partial \zeta}\right)
\end{equation}
and by 
\begin{equation}\label{antimm}
\textbf{J} (\vartheta, \zeta, \vartheta^+, \zeta^+) = \left(\begin{array}{cc}
- \vartheta \zeta^+ & \vartheta\vartheta^+\\
-\zeta\zeta^+ & \zeta\vartheta^+
\end{array}\right) ,
\end{equation}
respectively.
In block matrix notation an element $g$ of the group 
$$U(2,2):= \{g\in GL(4, \mathbb{C}): g^+ \phi g = \phi \}$$
has form
\begin{equation}
g=\left(\begin{array}{cc}
A & B\\
C & D
\end{array}\right) 
\end{equation}
 where $A,B,C,D \in \mbox{Mat}_{2\times 2}(\mathbb{C})$ satisfy
\begin{align}\nonumber
A^+C & = C^+A,\\
\label{anticon}
 D^+B & =B^+D, \\
\nonumber
A^+D & - C^+B = \sigma_0.
\end{align}

The Lie algebra $\textbf{u}(2,2)$ of $U(2,2)$ consists of the elements 
\begin{equation}\label{xdef}
\mathcal{X} = \left(\begin{array}{cc}
\lambda + i \epsilon & \tau \\
\alpha & -\lambda + i \epsilon 
\end{array}\right), 
\end{equation}
where $\alpha, \lambda, \epsilon ,\tau \in H(2)$. Decomposing $\alpha = \alpha^\mu \sigma_\mu, \lambda = \lambda^\mu \sigma_\mu , \epsilon = \epsilon^\mu \sigma_\mu $ and $\tau=\tau^\mu \sigma_\mu $ in the basis \eqref{pauli} we obtain the basis of the Lie algebra $\textbf{u}(2,2)$:
\begin{equation}\label{49}
\begin{array}{rl}
\mathcal{J}_\mu:=\frac{1}{2}\left(\begin{array}{cc}
i\sigma_\mu & 0 \\
0 & i\sigma_\mu
\end{array}\right), & \quad \mathcal{L}_\mu:=\frac{1}{2}\left(\begin{array}{cc}
\sigma_\mu & 0 \\
0 & -\sigma_\mu
\end{array}\right),\\
\mathcal{T}_\mu:=  \left(\begin{array}{cc}
0 & \sigma_\mu \\
0 & 0 
\end{array}\right), & \quad \mathcal{A}_\mu:=  \left(\begin{array}{cc}
0 & 0 \\
\sigma_\mu & 0 
\end{array}\right),
\end{array}
\end{equation}
where $\mu =0,1,2,3$.

The basis $\mathcal{J}^*_\mu, \mathcal{L}^*_\mu, \mathcal{T}^*_\mu, \mathcal{A}^*_\mu$ of $\textbf{u}(2,2)$ dual to the basis \eqref{49} with respect to the pairing \eqref{linfun} is given by 
\begin{equation}\label{410}
\mathcal{J}^*_\mu = \mathcal{J}_\mu, \quad \mathcal{L}^*_\mu = \mathcal{L}_\mu, \quad \mathcal{T}^*_\mu = \frac{1}{2}\mathcal{A}_\mu, \quad \mathcal{A}^*_\mu = \frac{1}{2}\mathcal{T}_\mu. 
\end{equation}
Therefore, we have a vector space isomorphism $\mathfrak{b}:\textbf{u}(2,2)\stackrel{\sim}{\rightarrow} \textbf{u}(2,2)$ defined by 
\begin{equation}\label{biso}
\begin{array}{ll}
\mathfrak{b} (\mathcal{J}_\mu) = \mathcal{J}^*_\mu= \mathcal{J}_\mu, & \quad  \mathfrak{b}(\mathcal{L}_\mu) = \mathcal{L}^*_\mu=\mathcal{L}_\mu ,\\
 \mathfrak{b}(\mathcal{T}_\mu)=\mathcal{T}^*_\mu=\frac{1}{2}\mathcal{A}_\mu,& \quad  \mathfrak{b}(\mathcal{A}_\mu)=\mathcal{A}^*_\mu =\frac{1}{2}\mathcal{T}_\mu.
\end{array} 
\end{equation}

In order to explain the terminology used in subsequent we identify the Minkowski space $(M_{1,3}, \eta )$ with $(H(2), \mbox{ det})$ by
\begin{equation}\label{413}
M_{1,3} \ni (x^\mu) \mapsto X:= x^\mu \sigma_\mu \in H(2),
\end{equation}
where $\eta$ is Minkowski metric tensor, i.e. 
\begin{equation}
(x^0)^2-(x^{1})^2 - (x^{2})^2 - (x^{3})^2 = \eta_{\mu\nu } x^\mu x^\nu = \mbox{det}(X). 
\end{equation}
Let us stress that in this paper we use alternatively the upper $V^\nu$ and the lower $V_\mu = \eta_{\mu\nu }V^\nu$ indices of $4$-vector $V$.

Now we distinguish the Lie subalgebras of $\textbf{u}(2,2)$ important from the physical point of view as well as essential for our further considerations. 
\begin{prop}\label{prop:41}
One has the decomposition of $\textbf{u}(2,2)$  on the direct sum   
\begin{equation}\label{412}
\textbf{u}(2,2) = \textbf{e}(3) \oplus \textbf{a}(2)\oplus \mathbb{R}\mathcal{L}_0 \oplus \textbf{l}(3)  \oplus \textbf{c}(4),
\end{equation}
of vector subspaces such that:
\begin{itemize}
\item[\textbf{(i)}] $\textbf{e}(3)$  is the Euclidean Lie algebra spanned by the elements $\mathcal{J}_k,  \mathcal{T}_k$, $k=1,2,3$,  of the basis \eqref{49}, which satisfy the commutation relations 
\begin{equation}\label{411}
[\mathcal{J}_k, \mathcal{J}_l] = \epsilon_{klm}\mathcal{J}_m, \quad [\mathcal{J}_k, \mathcal{T}_l] = \epsilon_{klm}\mathcal{T}_m, \quad [\mathcal{T}_k, \mathcal{T}_l] = 0.
\end{equation}
\item[\textbf{(ii)}] $\textbf{a}(2) := \mathbb{R}\mathcal{J}_0\oplus \mathbb{R} \mathcal{T}_0$ is an abelian  Lie algebra, i.e. 
\begin{equation}\label{414}
[\mathcal{J}_0, \mathcal{T}_0]=0.
\end{equation} 
\item[\textbf{(iii)}] $\textbf{e}(3)\oplus \textbf{a}(2)$ is the Lie algebra, which is a extension of $\textbf{e}(3)$ by $\textbf{a}(2)$, i.e. except of \eqref{411} and \eqref{414} one has
\begin{equation}
[\mathcal{J}_0, \mathcal{J}_k]=[\mathcal{J}_0, \mathcal{T}_k]=[\mathcal{T}_0, \mathcal{J}_k] = [\mathcal{T}_0, \mathcal{T}_k]=0. 
\end{equation}
\item[\textbf{(iv)}] $\textbf{p}(1,3) := \textbf{e}(3) \oplus \textbf{l}(3) \oplus \mathbb{R}\mathcal{T}_0$, where $\textbf{l}(3)$ is spanned by $\mathcal{L}_k$, $k=1,2,3$, is the Poincare Lie algebra, i.e. except of \eqref{411} the generators $\mathcal{J}_k, \mathcal{T}_\mu, \mathcal{L}_m $ satisfy 
\begin{equation}
[\mathcal{J}_k, \mathcal{L}_l ] = \epsilon_{klm}\mathcal{L}_m, \quad [\mathcal{T}_k, \mathcal{L}_l] = \epsilon_{klm}\mathcal{T}_m, \quad [\mathcal{L}_k, \mathcal{L}_l]= \epsilon_{klm}\mathcal{J}_m,
\end{equation}
\begin{equation}
 [\mathcal{T}_0, \mathcal{L}_k]=-\mathcal{T}_k.
\end{equation}
\item[\textbf{(v)}] $\textbf{c}(4)$ is the abelian  Lie algebra  spanned by $\mathcal{A}_\mu, \mu =0,1,2,3$, i.e.
\begin{equation}
[\mathcal{A}_\mu, \mathcal{A}_\nu ]=0. 
\end{equation}
\end{itemize}
\end{prop}
\begin{proof}
All statements of the proposition are established by straightforward verification. 
\end{proof}

Now let us define 
\begin{equation}\label{eq:420}
J= J^\mu \sigma_\mu , \quad \Gamma = \Gamma^\mu \sigma_\mu, \quad L= L^\mu \sigma_\mu, \quad K = K^\mu\sigma_\mu  \in H(2),
\end{equation}
where $ J^\mu, \Gamma^\mu, L^\mu, K^\mu \in \mathbb{R}$ are the coordinates 
\begin{equation}
\rho = J^\mu \mathcal{J}^*_\mu  + \Gamma^\mu \mathcal{T}_\mu^* + L^\mu \mathcal{L}_\mu ^* + K^\mu \mathcal{A}_\mu^*
\end{equation}
of $\rho\in \textbf{u}(2,2)$ in the dual basis \eqref{410}. Then, in the matrix block notation, $\rho$ assumes the following form
\begin{equation}\label{roblock}
\rho = \frac{1}{2} \left(\begin{array}{cc}
L+iJ & K \\
\Gamma & -L+iJ 
\end{array}\right). 
\end{equation}
Applying the isomorphism $\mathfrak{b}:\textbf{u}(2,2) \stackrel{\sim}{\rightarrow}\textbf{u}(2,2)$ defined in \eqref{biso} to the decomposition \eqref{412} we obtain the decomposition
\begin{equation}\label{419}
\textbf{u}(2,2) = \textbf{e}(3)^*\oplus \textbf{a}(2)^* \oplus (\mathbb{R}\mathcal{L}_0)^* \oplus \textbf{l}(3)^* \oplus \textbf{c}(4)^*
\end{equation}
of $\textbf{u}(2,2)$ on the subspaces dual to the Lie subalgebras mentioned in Proposition \ref{prop:41}, where 
\begin{align}
\label{420}
\textbf{e}(3)^* = & \mathfrak{b}(\textbf{e}(3)) = \left\{\frac{1}{2}\left(\begin{array}{cc}
i \vec{J}\cdot \vec{\sigma} & 0 \\
 \vec{\Gamma}\cdot \vec{\sigma} & i \vec{J}\cdot \vec{\sigma}
\end{array}\right): \vec{\Gamma}, \vec{J} \in \mathbb{R}^3 \right\}, \\
\label{421}
\textbf{a}(2)^* = & \mathfrak{b}(\textbf{a}(2)) = \left\{\frac{1}{2}\left(\begin{array}{cc}
i J^0\sigma_0 & 0 \\
 \Gamma^0 \sigma_0 & i J^0\sigma_0
\end{array}\right): \Gamma^0, J^0 \in \mathbb{R} \right\}, \\
\label{422}
(\textbf{e}(3)\oplus \textbf{a}(2))^* = & \mathfrak{b}(\textbf{e}(3)\oplus \textbf{a}(2)) = \left\{\frac{1}{2}\left(\begin{array}{cc}
i J & 0 \\
 \Gamma & i J
\end{array}\right): \Gamma, J \in H(2)\right\}, \\ 
(\mathbb{R}\mathcal{L}_0\oplus \textbf{l}(3))^* = & \mathfrak{b}(\mathbb{R}\mathcal{L}_0\oplus \textbf{l}(3)) = \left\{\frac{1}{2}\left(\begin{array}{cc}
L & 0 \\
0 & -L
\end{array}\right): L \in H(2)\right\}, \\ 
\textbf{c}(4)^* = & \mathfrak{b}(\textbf{c}(4)) = \left\{\frac{1}{2}\left(\begin{array}{cc}
0 & K \\
0 & 0
\end{array}\right): K \in H(2)\right\}. 
\end{align}
Therefore, we define the momentum maps $\textbf{J}_u: \mathbb{T} \to (\textbf{e}(3)\oplus \textbf{a}(2))^*$, $\textbf{J}_e : \mathbb{T} \to \textbf{e}(3)^*$ and  $\textbf{J}_{a} : \mathbb{T} \to \textbf{a}(2)^*$  as the superpositions $\textbf{J}_u= \pi_u \circ \textbf{J}$, $\textbf{J}_e := \pi_e \circ \textbf{J}$ and  $\textbf{J}_{a} = \pi_{a}\circ \textbf{J}$   of the twistor momentum map $\textbf{J}: \mathbb{T} \to \textbf{u}(2,2)$, see \eqref{antimm}, with the projections $ \pi_u :\textbf{u}(2,2)\to (\textbf{e}(3)\oplus\textbf{a}(2))^*$, $\pi_e: \textbf{u}(2,2)\to \textbf{e}(3)^*$ and  $\pi_{a} : \textbf{u}(2,2)\to \textbf{a}(2)^*$ of  the Lie algebra $\textbf{u}(2,2)$ on the suitable components of the decomposition \eqref{419}. Thus, comparing \eqref{antimm} with \eqref{420}, \eqref{421} and \eqref{422}, respectively, we obtain 
\begin{align}
\textbf{J}_{u} (\vartheta , \zeta ) & = \frac{1}{2}\left(\begin{array}{cc}
i J(\vartheta, \zeta) & 0 \\
 \Gamma (\vartheta , \zeta) & i J (\vartheta , \zeta ) 
\end{array}\right), \\
\textbf{J}_e (\vartheta , \zeta ) & = \frac{1}{2}\left(\begin{array}{cc}
i \vec{J}(\vartheta, \zeta ) \cdot \vec{\sigma} & 0 \\
 \vec{\Gamma}(\vartheta , \zeta ) \cdot \vec{\sigma} & i \vec{J}(\vartheta, \zeta ) \cdot \vec{\sigma}
\end{array}\right)\\
\label{eqja}
\textbf{J}_{a} (\vartheta , \zeta ) & = \frac{1}{2}\left(\begin{array}{cc}
i J^0(\vartheta, \zeta ) \sigma_0 & 0 \\
 \Gamma^0(\vartheta , \zeta ) \sigma_0 & i J^0(\vartheta, \zeta ) \sigma_0
\end{array}\right)
\end{align}
where 
\begin{align}\label{428i}
J(\vartheta , \zeta ) & = i(\vartheta \zeta^+ - \zeta \vartheta^+),\\
\label{428}
 \vec{J}(\vartheta , \zeta ) & = \frac{i}{2}( \zeta^+ \vec{\sigma} \vartheta- \vartheta^+\vec{\sigma} \zeta ), \\
\label{428j}
J^0 (\vartheta , \zeta ) & = \frac{i}{2} (\zeta^+\vartheta  - \vartheta^+\zeta ), 
\end{align}
and 
\begin{align}\label{429i}
\Gamma (\vartheta , \zeta ) & = -2 \zeta \zeta^+, \\
\label{429}
\vec{\Gamma }(\vartheta , \zeta ) & = - \zeta^+\vec{\sigma}\zeta , \\
\label{429j}
\Gamma^0 (\vartheta, \zeta ) & = - \zeta^+\zeta . 
\end{align}

We see from \eqref{420}, \eqref{421} and \eqref{422} that one has the following isomorphisms
\begin{equation}\label{iso432}
\begin{array}{c}
(\textbf{e}(3)\oplus \textbf{a}(2))^* \cong H(2)\times H(2) \cong \mathbb{R}^4 \times \mathbb{R}^4 \mbox{ and } \\
\textbf{e}(3)^* \cong \mathbb{R}^3 \times \mathbb{R}^3 , \qquad \textbf{a}(2)^* \cong \mathbb{R}\times \mathbb{R}  
\end{array}
\end{equation}
of the vector spaces. 
So, in subsequent we will consider the momentum maps $\textbf{J}_e : \mathbb{T} \to \textbf{e}(3)^*$, $\textbf{J}_{a(2)} : \mathbb{T} \to \textbf{a}(2)^*$ and $\textbf{J}_u: \mathbb{T} \to (\textbf{e}(3)\oplus \textbf{a}(2))^*$ as the maps 
\begin{align}
\label{eq:ju}
\textbf{J}_u: \mathbb{T}\ni  \left(\begin{array}{c}
\vartheta \\
 \zeta \end{array}\right) \mapsto \textbf{J}_{u} (\vartheta , \zeta ) & =\left(\begin{array}{c} 
J(\vartheta , \zeta)\\
\Gamma \vartheta, \zeta)
 \end{array}\right)  \in H(2)\times H(2)\\
\label{eq:je}
\textbf{J}_e: \mathbb{T}\ni  \left(\begin{array}{c}
\vartheta \\
 \zeta \end{array}\right) \mapsto \textbf{J}_{e} (\vartheta , \zeta )& = \left(\begin{array}{c} 
\vec{J}(\vartheta , \zeta)\\
\vec{\Gamma}( \vartheta, \zeta)
\end{array}\right)\in \mathbb{R}^3\times \mathbb{R}^3\\
\label{435}
\textbf{J}_{a}: \mathbb{T}\ni  \left(\begin{array}{c}
\vartheta \\
 \zeta \end{array}\right) \mapsto \textbf{J}_{a(2)} (\vartheta , \zeta )& =\left(\begin{array}{c} 
J^0(\vartheta , \zeta)\\
\Gamma^0 ( \vartheta, \zeta ) 
\end{array}\right) \in \mathbb{R}\times \mathbb{R}
\end{align}
of $\mathbb{T}$ into $\mathbb{R}^3\times \mathbb{R}^3$, $\mathbb{R}\times \mathbb{R}$ and $H(2)\times H(2)$, respectively.

Consistently with the decomposition \eqref{412}, we present below the Lie-Poisson brackets, see \eqref{lpnawias}, of the coordinates functions $J_\mu = \eta_{\mu\nu}J^\nu $, $\Gamma_\mu = \eta_{\mu\nu }\Gamma^\nu , L_\mu =\eta_{\mu\nu}L^\nu, K_\mu = \eta_{\mu\nu }L^\nu \in C^\infty (\textbf{u}(2,2), \mathbb{R})$: 
\begin{itemize}
\item[\textbf{(i)}] Lie-Poisson brackets of the coordinates functions $J_k,  \Gamma_k$, $k=1,2,3$, of $\textbf{e}(3)^*$ are  as in \eqref{LPBjg}.\\
\item[\textbf{(ii)}] Lie-Poisson brackets of the coordinates functions $J_\mu, \Gamma_\mu$, $\mu=0,1,2,3$, of $ (\textbf{e}(3)\oplus\textbf{a}(2))^*$ are \eqref{LPBjg} and 
\begin{equation}\label{bracket2}
\{J_0 , J_l \}_{_{LP}}  = \{J_0 , \Gamma_l \}_{_{LP}}= \{J_0, \Gamma_0\} = \{\Gamma_0, J_l \}_{_{LP}}= \{\Gamma_0, \Gamma_l \}_{_{LP}}=0,\\
\end{equation}
\item[\textbf{(iii)}] Lie-Poisson brackets of the coordinates functions $L_\mu$, $\mu = 0,1,2,3$, of $(\mathbb{R}\mathcal{L}_0 \oplus \textbf{l}(3))^*$ are
\begin{equation}
\{J_k , L_l \}_{_{LP}}  = \epsilon_{klm}L_m, \quad  \{\Gamma_k , L_l \}_{_{LP}}  = i\epsilon_{klm}\Gamma_m, \quad \{L_k , L_l \}_{_{LP}}  = \epsilon_{klm}J_m, 
\end{equation}
\begin{equation}
\{\Gamma_0 , L_\mu \}_{_{LP}}= \Gamma_\mu , \quad \{J_0 , L_\mu \}_{_{LP}}= \{L_0 , L_k \}_{_{LP}}= \{L_0 , J_k\}_{_{LP}}=0,
\end{equation}
\item[\textbf{(iv)}] Lie-Poisson brackets of the coordinates functions $K_\mu $, $\mu =0,1,2,3$, of $\textbf{c}(4)^*$ are 
\begin{equation}
\begin{array}{ll}
\{K_\mu , K_\nu \}_{_{LP}} = 0, \quad & \{L_k , K_l \}_{_{LP}} = -i\epsilon_{klm}K_m, \\
 \{\Gamma_k , K_l \}_{_{LP}}= \frac{1}{2}\epsilon_{klm}J_m, \quad & \{J_k , K_l \}_{_{LP}}   = \epsilon_{klm}K_m,
\end{array}
\end{equation}
\begin{equation}
 \{J_0 , K_l \}_{_{LP}}  =  \{\Gamma_0 , K_l \}_{_{LP}} = \{L_0 , K_l \}_{_{LP}} = 0 .
\end{equation}
\end{itemize}

We also consider the subgroups $P(2,2)$ and $C(4)$ of $U(2,2)$ defined as follows
\begin{equation}\label{antig}
P(2,2):= \left\{\left(\begin{array}{cc}
A & AT\\
0 & (A^+)^{-1}
\end{array}\right) \in U(2,2): A \in GL(2, \mathbb{C}), T \in H(2)\right\}
\end{equation}
and 
\begin{equation}
C(4) := \left\{\left(\begin{array}{cc}
\sigma_0 & 0 \\
C & \sigma_0 
\end{array}\right) : C\in H(2)\right\}. 
\end{equation}
These subgroups act on $(H(2), \mbox{det})$ in the following way
\begin{equation}
\sigma_g X = A(X+T)A^+,
\end{equation}
where $g\in P(2,2)$, and 
\begin{equation}\label{4.24}
\sigma_gX = X(CX + \sigma_0)^{-1}, 
\end{equation}
where $g \in C(4)$. The Lie algebra of $C(4)$ is the commutative Lie algebra $\textbf{c}(4)$ mentioned in the point (v) of Proposition \ref{prop:41}. Note here that \eqref{4.24} is defined only if $\mbox{det}(CX+\sigma_0) \neq 0$.  

The subgroup $P(2,2)$ is isomorphic $P(2,2) \cong GL(2, \mathbb{C})\ltimes H(2)$ with semidirect product   of the linear group $GL(2, \mathbb{C})$ with  $(H(2), +)$, i.e. the product of $(A_1, T_1), (A_2, T_2)\in  GL(2, \mathbb{C})\ltimes H(2)$ is defined by
\begin{equation}
(A_1, T_1)\cdot (A_2, T_2) := (A_1A_2, T_2 + A_2^{-1} T_1 (A_2^+)^{-1}). 
\end{equation}

Let us now  distinguish the subgroups of $GL(2, \mathbb{C})\ltimes H(2)$ corresponding to the Lie subalgebras mentioned in the points (i)-(iv) of Proposition \ref{prop:41}: 
\begin{itemize}
\item[\textbf{(i)}] The double covering $\tilde{E}(3)$ of the Euclidean group $E(3)$ defined by $\tilde{E}(3) :=SU(2)\ltimes H_0(2)$, where $ (A, T) \in SU(2)\ltimes H_0(2)$ iff $ AA^* = \sigma_0, \mbox{det}(A)=1$ and $T\in H_0(2):=\{T\in H(2): \mbox{Tr}(T) =0\}$. The Lie algbera of $\tilde{E}(3)$ is $\textbf{e}(3)$.
\item[\textbf{(ii)}] The group $\tilde{A}(2)$ consists of $ (A,T)= (e^{it} \sigma_0,  s \sigma_0)$, where $ t,s \in \mathbb{R}$. The Lie algbera of $\tilde{A}(2)\cong U(1)\times \mathbb{R}$ is $\textbf{a}(2)$.
\item[\textbf{(iii)}] The group $U(2)\ltimes H(2)$, consists of $(A, T) \in GL(2, \mathbb{C})\ltimes H(2)$ such that $ AA^* = \sigma_0$. The Lie algebra of $U(2)\ltimes H(2)$ is $\textbf{e}(3)\oplus \textbf{a}(2)$.
\item[\textbf{(iv)}] The double covering $\widetilde{P(1,3)} \cong SL(2, \mathbb{C})\ltimes H(2)$ of the Poincare group $P(1,3)$, i.e.  $(A, T ) \in SL(2, \mathbb{C})\ltimes H(2) $ iff $ \mbox{det}(A) =1$. The Lie algbera of $\widetilde{P(1,3)}$ is $\textbf{p}(1,3)$.
\end{itemize}

 The inclusions between these groups are presented in the diagram below 
\begin{equation*}\label{diagram}
\begin{tikzcd}
                             &                                   & {GL(2, \mathbb{C})\ltimes H(2)}                        &                                                   \\
\tilde{A}(2) \arrow[hook]{r} & U(2)\ltimes H(2) \arrow[hook]{ru} &                                                        & {SL(2,\mathbb{C})\ltimes H(2) } \arrow[swap,hook]{lu} \\
                             &                                   & SU(2)\ltimes H_0(2) \arrow[hook]{lu} \arrow[swap,hook]{ru} &                                                  
\end{tikzcd}
\end{equation*}

The action $\Sigma_g: \mathbb{T} \to \mathbb{T}$ of $g= \left(\begin{array}{cc}
A & AT\\
0 & A
\end{array}\right)\in U(2)\ltimes H(2)$ on $\left(\begin{array}{c}
\vartheta \\
\zeta 
\end{array}\right)\in\mathbb{T}$ is given by 
\begin{equation}\label{sigmaaction}
\Sigma_g \left(\begin{array}{c}
\vartheta \\
\zeta 
\end{array}\right) = \left(\begin{array}{c}
A(\vartheta + T \zeta )\\
A\zeta 
\end{array}\right).
\end{equation}
The coadjoint action $\tilde{\mbox{Ad}}^*_{g^{-1}}: \textbf{e}(3)^*\oplus \textbf{a}(2)^*\to \textbf{e}(3)^*\oplus \textbf{a}(2)^*$ of $g=(A, T) \in U(2)\ltimes H(2)$ on $H(2)\times H(2)\cong \textbf{e}(3)^*\oplus \textbf{a}(2)^*$ assumes the form 
\begin{equation}\label{adaction}
\tilde{\mbox{Ad}}^*_{g^{-1}} (\Gamma , J) = ( A(J+ \frac{i}{2}[\Gamma, T])A^+, A\Gamma A^+).
\end{equation}
The momentum map $\textbf{J}_u:\widetilde{\mathbb{T}}\to H(2)\times H(2)$, see \eqref{428}, \eqref{429} and \eqref{eq:ju}, is an equivariant map with respect to the actions \eqref{sigmaaction} and \eqref{adaction}, i.e. 
\begin{equation}\label{452}
\begin{tikzcd}
\mathbb{T} \arrow{r}{\Sigma_g} \arrow{d}{\textbf{J}_u} & \mathbb{T} \arrow{d}{\textbf{J}_u} \\
H(2)\times H(2) \arrow{r}{\tilde{Ad}_{g^{-1}}^*}                  & H(2)\times H(2)                     
\end{tikzcd}
\end{equation}
for arbitrary $g=(A,T)\in U(2)\ltimes H(2)$. The decomposition
\begin{equation}
H(2)\times H(2) = (H_0(2)\times H_0(2))\oplus (\mathbb{R}\sigma_0\times \mathbb{R}\sigma_0 )
\end{equation}
of $H(2)\times H(2)$ corresponding to the decomposition $\textbf{e}(3)^*\oplus \textbf{a}(2)^*$ through the isomorphism \eqref{iso432} is invariant with respect to the action \eqref{adaction}. It is also easy to see that  the action \eqref{adaction} of $U(2)\ltimes H(2)$ on $\mathbb{R}\sigma_0\times \mathbb{R}\sigma_0 \cong \textbf{a}(2)^*$ is trivial. From the above we conclude that the momentum maps $\textbf{J}_{a}:\mathbb{T}\to \mathbb{R}\times \mathbb{R}$ and $\textbf{J}_e:\mathbb{T}\to \mathbb{R}^3 \times \mathbb{R}^3 $ are $U(2)\ltimes H(2)$-equivariant maps and thus, they are $\tilde{A}(2)$ and $\tilde{E}(3)\cong SU(2)\ltimes H_0(2)$ equivariant maps, respectively. Using isomorphism $H_0(2)\times H_0(2)\cong \mathbb{R}^3 \times \mathbb{R}^3$ we find from \eqref{adaction} that the coadjoint action $\mbox{Ad}^*_{g^{-1}}: \mathbb{R}^3 \times \mathbb{R}^3 \to \mathbb{R}^3 \times \mathbb{R}^3$ of $g = (A, \vec{T}\cdot \vec{\sigma})\in SU(2)\ltimes H_0(2)\cong \tilde{E}(3)$ on $\mathbb{R}^3 \times \mathbb{R}^3$ is the following 
\begin{equation}\label{eaction}
\mbox{Ad}^*_{g^{-1}} \left(\begin{array}{c}
\vec{J} \\
\vec{\Gamma}
\end{array}\right) = \left(\begin{array}{c}
O(\vec{J} + \vec{T}\times \vec{\Gamma})\\
O\vec{\Gamma}\\
\end{array}\right),
\end{equation}
where the matrix elements of the rotation $O\in SO(3)$ depend on $A\in SU(2)$ by 
\begin{equation}\label{ortmatrix}
O_{kl} = \frac{1}{2} \mbox{Tr}(\sigma_k A^+ \sigma_l A). 
\end{equation}
Let us note here that the dependence \eqref{ortmatrix} defines the double covering group epimorphism $SU(2)\ltimes H_0(2) \cong \tilde{E}(3) \to E(3)$ mentioned above.

Now, basing on the above facts, we will present statements which describe the properties of the momentum maps $\textbf{J}_u:\mathbb{T}\to \mathbb{R}^4\times \mathbb{R}^4$, $\textbf{J}_e:\mathbb{T}\to \mathbb{R}^3\times \mathbb{R}^3$ and $\textbf{J}_a:\mathbb{T}\to \mathbb{R}\times \mathbb{R}$. We will omit proofs of these statements having in mind their rather technical character. 
 \begin{prop}\label{prop:41}
\begin{itemize}
\item[\textbf{(i)}] The momentum maps $\textbf{J}_e:\mathbb{T}\to \mathbb{R}^3\times \mathbb{R}^3$, $\textbf{J}_{a}:\mathbb{T}\to \mathbb{R}\times \mathbb{R}$ and $\textbf{J}_u:\mathbb{T}\to \mathbb{R}^4\times \mathbb{R}^4 $ are respectively $\tilde{E}(3)$, $\tilde{A}(2)$ and $U(2)\ltimes H(2)$ equivariant Poisson maps of the twistor symplectic space $(\mathbb{T}, \Omega)$ into the respective Lie-Poisson spaces, where: the Poisson bracket on $\mathbb{R}^3\times \mathbb{R}^3 $ is given by \eqref{LPBjg}; from \eqref{bracket2} it follows that the bracket on $\mathbb{R}\times \mathbb{R}$  is trivial i.e. $\{f, g \} = 0 $ for any $f,g \in C^\infty(\mathbb{R}, \mathbb{R})$;  the Poisson bracket on $\mathbb{R}^4\times\mathbb{R}^4\cong H(2)\times H(2)$ is given by \eqref{LPBjg} and \eqref{bracket2}. 
\item[\textbf{(ii)}] One has $( J, \Gamma ) \in \textbf{J}_u(\mathbb{T})$ if and only if 
\begin{equation}\label{conobraz}
\mbox{det} (\Gamma) =0, \quad \mbox{Tr} (\Gamma )\leq 0 \mbox{ and } \mbox{Tr}(\Gamma J) = \mbox{Tr}(\Gamma)\mbox{Tr}(J).
\end{equation}
Writing $( J, \Gamma) \in H(2)\times H(2)$ in the coordinates defined in \eqref{eq:420}, we could express the conditions (\ref{conobraz}) in the form 
\begin{equation}\label{conobraz2}
(\Gamma^0)^2 - \vec{\Gamma}^2 =0, \quad \Gamma^0 \leq 0, \quad \Gamma^0 J^0 - \vec{\Gamma}\cdot \vec{J} =0.
\end{equation}
\item[\textbf{(iii)}] The null-levels $\textbf{J}_e^{-1}(0,0)$, $\textbf{J}_{a}^{-1}(0,0)$ and $ \textbf{J}_u^{-1}(0,0)$ of the momentum maps $\textbf{J}_e, \textbf{J}_a$ and $\textbf{J}_u$ are equal to the linear subspace  $\left\{\left(\begin{array}{c}
\vartheta \\
\zeta 
\end{array}\right) \in \mathbb{T} : \zeta =0\right\}$ of $ \mathbb{T}$, which is a Lagrangian submanifold  of $(\mathbb{T}, \Omega)$. 
\end{itemize}
\end{prop}

Let us consider the open dense subset $\widetilde{\mathbb{T}} := \mathbb{T}\backslash \left\{\left(\begin{array}{c} 
\vartheta\\
0
\end{array}\right) : \vartheta \in \mathbb{C}^2\right\}$ of $\mathbb{T}$. We see from \eqref{conobraz2} that the image $\textbf{J}_u (\widetilde{\mathbb{T}})$ of $\widetilde{\mathbb{T}}$ is diffeomorphic  $\textbf{J}_u (\widetilde{\mathbb{T}}) \cong TC_{-3}$ with the tangent bundle of the lower half $C_{-}:= \{ \Gamma \in H(2) : \mbox{det}(\Gamma) =0 \mbox{ and } \mbox{Tr}(\Gamma)<0\}$ of the cone in $H(2)$. The map 
\begin{equation}
 \mathbb{R}^3 \times \dot{\mathbb{R}}^3\ni ( \vec{J}, \vec{\Gamma}) \mapsto \left( \frac{-1}{\sqrt{\vec{\Gamma}^2}} \vec{\Gamma} \cdot \vec{J} + \vec{J}\cdot \vec{\sigma}, -\sqrt{\vec{\Gamma}^2}\sigma_0 + \vec{\Gamma}\cdot \vec{\sigma}\right) \in \textbf{J}_u(\widetilde{\mathbb{T}})
\end{equation}
gives an isomorphism $TC_{-} \cong  \mathbb{R}^3\times \dot{\mathbb{R}}^3$ between the tangent bundle $TC_{-} \to C_{-}$ and the trivial vector bundle $ \mathbb{R}^3\times \dot{\mathbb{R}}^3\to \dot{\mathbb{R}}^3$, where $\dot{\mathbb{R}}^3 := \mathbb{R}^3 \backslash \{0\}$. 
Let us note  additionally that $\textbf{J}_e(\widetilde{\mathbb{T}}) \cong  \mathbb{R}^3\times \dot{\mathbb{R}}^3$ and $\textbf{J}_{a}(\widetilde{\mathbb{T}}) \cong \mathbb{R}\times \mathbb{R}_- $, where $\mathbb{R}_-:= ]-\infty, 0[$. In subsequent we will treat $\mathbb{R}\times \mathbb{R}_-$ as a Poisson manifold with trivial Poisson structure.

It is useful to define the following surjective submersion $\textbf{J}_{a,e} :  \mathbb{R}^3\times \dot{\mathbb{R}}^3 \to \mathbb{R}\times \mathbb{R}_-$ by 
\begin{equation}\label{jae}
\textbf{J}_{a,e} (\vec{J}, \vec{\Gamma}) := \left(-\frac{1}{\sqrt{\vec{\Gamma}^2}} \vec{\Gamma}\cdot \vec{J}, - \sqrt{\vec{\Gamma}^2}\right). 
\end{equation}

In next proposition we will use such notions as a full and complete symplectic realization of a Poisson manifold and a symplectic dual pair, see for example \cite{We,W}. For their definition see also Appendix \ref{B}.

\begin{prop}\label{prop:2}
\begin{itemize}
\item[\textbf{(i)}] One has the following diagram
\begin{equation}\label{diagram66}
\begin{tikzcd}
                                    & \widetilde{\mathbb{T}} \arrow{ld}[swap]{\pi_{\tilde{E}(3)}} \arrow{rd}{\pi_{\tilde{A}(2)}} &                                                                                      \\
\widetilde{\mathbb{T}}/\tilde{E}(3) &                                                                                          & \widetilde{\mathbb{T}}/\tilde{A}(2) \arrow{ll}{\pi_{\tilde{E}(3), \tilde{A}(2)}}
\end{tikzcd},
\end{equation}
whose arrows are full and complete Poisson maps:\\
1) the quotient map $\pi_{\tilde{A}(2)}: \widetilde{\mathbb{T}}\to \widetilde{\mathbb{T}}/\tilde{A}(2)$ defined by the regular action \eqref{sigmaaction} of $\tilde{A}(2)$ on $\widetilde{\mathbb{T}}$;\\
2) the quotient map $\pi_{\tilde{E}(3)}: \widetilde{\mathbb{T}}\to \widetilde{\mathbb{T}}/\tilde{E}(3)$ defined by the regular action of $\tilde{E}(3)$ on $\widetilde{\mathbb{T}}$;\\
3) the quotient map $\pi_{\tilde{E}(3), \tilde{A}(2)}: \widetilde{\mathbb{T}}\to (\widetilde{\mathbb{T}}/\tilde{A}(2))/\tilde{E}(3)$ defined by the regular action of $\tilde{E}(3)$ on $\widetilde{\mathbb{T}}/\tilde{A}(2)$.\\
The Poisson brackets on $C^\infty(\widetilde{\mathbb{T}}/\tilde{A}(2), \mathbb{R})\cong C^\infty_{\tilde{A}(2)}(\widetilde{\mathbb{T}}, \mathbb{R})$ and on  $C^\infty(\widetilde{\mathbb{T}}/\tilde{E}(3), \mathbb{R})\cong C^\infty_{\tilde{E}(3)}(\widetilde{\mathbb{T}}, \mathbb{R})$ are obtained by the restriction of the twistor Poisson bracket \eqref{antipb} to the Poisson subalgebras $C^\infty_{\tilde{A}(2)}(\widetilde{\mathbb{T}}, \mathbb{R})$ and $C^\infty_{\tilde{E}(3)}(\widetilde{\mathbb{T}}, \mathbb{R})$ of the $\tilde{A}(2)$ and $\tilde{E}(3)$ invariant functions. Let us stress that the bracket on $\widetilde{\mathbb{T}}/\tilde{E}(3)$ is the trivial one. 
\item[\textbf{(ii)}] The levels $\textbf{J}_{a}^{-1}(\textbf{J}_{a}(w))$ and $\textbf{J}_e^{-1}(\textbf{J}_e(w))$ of $\textbf{J}_{a}$ and $\textbf{J}_e$ through $w\in \widetilde{\mathbb{T}}$ are connected submanifolds of $\widetilde{\mathbb{T}}$.  The groups  $\tilde{E}(3)$ and $\tilde{A}(2)$  act in a free and transitive way on $\textbf{J}_{a}^{-1}(\textbf{J}_{a}(w))$ and $\textbf{J}_e^{-1}(\textbf{J}_e(w))$, respectively. 
\item[\textbf{(iii)}] The diagram 
\begin{equation}\label{diagram62}
\begin{tikzcd}
                & \widetilde{\mathbb{T}} \arrow{ld}[swap]{\textbf{J}_{a}} \arrow{rd}{\textbf{J}_e} &                 \\
 \mathbb{R}\times \mathbb{R}_- &                                                                               &  \mathbb{R}^3\times \dot{\mathbb{R}}^3\arrow{ll}{\textbf{J}_{a,e}}
\end{tikzcd},
\end{equation}
which arrows are full and complete Poisson maps equivariant with respect to $\tilde{E}(3)$, is isomorphic to the diagram \eqref{diagram66}. This means that arrows of this diagram correspond to the arrows of the diagram \eqref{diagram66} through the diffeomorphisms $\widetilde{\mathbb{T}}/\tilde{E}(3)\cong \mathbb{R}\times \mathbb{R}_-$ and $\widetilde{\mathbb{T}}/\tilde{A}(2)\cong \mathbb{R}^3\times \dot{\mathbb{R}}^3$.
\end{itemize}
\end{prop}
\begin{proof}
The proof of the above statements is established by straightforward verification, where completness of the maps in the diagram \eqref{diagram62} follows from its isomorphism to the diagram \eqref{diagram66} and  Proposition 6.6 in \cite{W}. 
\end{proof}
From the above proposition we conclude
\begin{prop}\label{cor:44}
\begin{itemize}
\item[\textbf{(i)}] The momentum maps $\textbf{J}_{a}: \widetilde{\mathbb{T}} \to  \mathbb{R}\times \mathbb{R}_- $, $\textbf{J}_e :\widetilde{\mathbb{T}} \to \mathbb{R}^3\times \dot{\mathbb{R}}^3 $ and $\textbf{J}_{a,e}: \mathbb{R}^3\times \dot{\mathbb{R}}^3\to \mathbb{R}\times \mathbb{R}_-$ from \eqref{diagram62} are full and complete symplectic realizations of Poisson manifolds $  \mathbb{R}\times \mathbb{R}_- $  and $ \mathbb{R}^3\times \dot{\mathbb{R}}^3 $. 
\item[\textbf{(ii)}]   The fibres of $\textbf{J}_{a}$ and $\textbf{J}_e$ are connected submanifolds of $\widetilde{\mathbb{T}}$ and are symplectically orthogonal  with respect to $\Omega$.  
\item[\textbf{(iii)}] The Poisson manifolds $  \mathbb{R}\times \mathbb{R}_- $  and $ \mathbb{R}^3\times \dot{\mathbb{R}}^3 $ form a symplectic dual pair. 
\item[\textbf{(iv)}] There is a one-to-one correspondence between the elements $(\mu, \nu) \in  \mathbb{R}\times \mathbb{R}_- = \textbf{J}_{a}(\widetilde{\mathbb{T}})$ and the symplectic leaves $\textbf{J}_e(\textbf{J}_{a}^{-1}(\mu, \nu))$ of $\textbf{e}(3)^*$ which are equal to the levels $\textbf{J}_{a,e}^{-1}(\mu, \nu)$ of $\textbf{J}_{a,e}$. 
\end{itemize}
\end{prop}
\begin{proof}
\textbf{(i)} Follows from (iii) of Proposition \ref{prop:2}.\\
\textbf{(ii)} Follows from connectness of the groups $\tilde{E}(3)$ and $\tilde{A}(2)$. \\
\textbf{(iii)} Is a consequence of (i) and (ii). \\
\textbf{(iv)} Follows from (iii) and Proposition 9.2 in \cite{W}. 
\end{proof}

Let us mention the relationship between the momentum map $\textbf{J}_{a,e}: \mathbb{R}^3\times \dot{\mathbb{R}}^3 \to  \mathbb{R}\times \mathbb{R}_- $ and the map $K: \mathbb{R}^3\times \dot{\mathbb{R}}^3 \to  \mathbb{R}\times \mathbb{R}_+$ defined by the Casimir functions \eqref{casimirs}, i.e. 
\begin{equation}\label{casfun}
K(\vec{J}, \vec{\Gamma}) := (K_1 (\vec{J}, \vec{\Gamma}), K_2(\vec{J}, \vec{\Gamma})) = (\vec{\Gamma}\cdot \vec{J}, \vec{\Gamma}^2). 
\end{equation}
Comparing \eqref{jae} with \eqref{casfun} we find that $K= \Delta \circ \textbf{J}_{a,e}$, where the diffeomorphism $\Delta: \mathbb{R}\times \mathbb{R}_-\to \mathbb{R}\times \mathbb{R}_+$ is defined by 
\begin{equation}\label{casfun2}
\left(\begin{array}{c}
c_1 \\
c_2 
\end{array}\right) = \Delta (\mu, \nu) := \left(\begin{array}{c}
\mu \nu  \\
\nu^2
\end{array}\right). 
\end{equation}
So, one can identify the levels $K^{-1}(c_1, c_2)$ of $K$ with the levels $\textbf{J}_{a,e}^{-1}(\mu, \nu)$ of $\textbf{J}_{a,e}$. The above implies that the levels $K^{-1}(c_1, c_2)$ are the coadjoint orbits of $E(3)$.

Ending this section we write the twistor symplectic form $\Omega$ given by \eqref{2form} and the Euclidean momentum map $\textbf{J}_e:\mathbb{T} \to \mathbb{R}^3\times \mathbb{R}^3\cong \textbf{e}(3)^*$ given by \eqref{eq:je} in the real coordinates $(q_0, q_1, q_2, q_3, \pi_0, \pi_1, \pi_2, \pi_3)\in \mathbb{R}^4\times \mathbb{R}^4$ dependent on the spinor coordinates $(\vartheta , \zeta )\in \mathbb{C}^2\times \mathbb{C}^2$ in the following way 
\begin{equation}\label{rcn}
\begin{array}{ll}
\left(\begin{array}{c}
q_0\\
q_1
\end{array}\right) = \frac{1}{\sqrt{2}}(\zeta + \bar\zeta), & \left(\begin{array}{c}
q_2\\
q_3
\end{array}\right) = \frac{-i}{\sqrt{2}}(\zeta - \bar\zeta),\\
\left(\begin{array}{c}
\pi_0\\
\pi_1
\end{array}\right) = \frac{1}{\sqrt{2}}(\vartheta + \bar\vartheta), & \left(\begin{array}{c}
\pi_2\\
\pi_3
\end{array}\right) = \frac{-i}{\sqrt{2}}(\vartheta - \bar\vartheta).
\end{array}
\end{equation}
Combining \eqref{rcn} with \eqref{2form} we find that
\begin{equation}\label{2formr}
\Omega = dq_0\wedge d\pi_0 +dq_1\wedge d\pi_1 + dq_2\wedge d\pi_2 + dq_3\wedge d\pi_3,
\end{equation}
i.e. $(q, \pi ) := (q_\mu, \pi_\mu )$, where $\mu =0,1,2,3$, are the canonical coordinates on $\mathbb{T} \cong \mathbb{R}^4\times \mathbb{R}^4\cong T^*\mathbb{R}^4$. The vector components  $\vec{J}:\mathbb{T}\to \mathbb{R}^3$ and $\vec{\Gamma}: \mathbb{T}\to \mathbb{R}^3$ of $\textbf{J}_e:  \mathbb{T}\to \mathbb{R}^3\times \mathbb{R}^3$ defined in \eqref{428} and \eqref{429} respectively, in these canonical coordinates assume the form 
\begin{align}\label{njot}
\vec{J}(q, \pi ) & = \frac{1}{2}\left(\begin{array}{c}
q_0\pi_3+q_1\pi_2-q_2\pi_1 - q_3\pi_0 \\
q_0\pi_1 -q_1\pi_0 +q_2\pi_3 -q_3\pi_2\\
q_0\pi_2 -q_1\pi_3-q_2\pi_0 +q_3\pi_1
\end{array}\right)\\
\label{ngamma}
\vec{\Gamma }(q, \pi ) & = \left(\begin{array}{c}
q_0q_1+q_2q_3\\
q_1q_2 -q_0q_3\\
\frac{1}{2} (q_0^2+q_1^2 - q_2^2-q_3^2)
\end{array}\right).
\end{align}
The components $J^0:\mathbb{T}\to \mathbb{R}$ and $\Gamma^0: \mathbb{T}\to \mathbb{R}$ of the momentum map $\textbf{J}_a: \mathbb{T} \to \mathbb{R}\times \mathbb{R} $ written in $(q, \pi)$ are 
\begin{align}\label{eq:469}
J^0 (q, \pi ) & = \frac{1}{2} (q_2\pi_0 + q_3\pi_1 - q_0\pi_2 - q_1\pi_3)\\
\Gamma^0 (q, \pi ) & = -\frac{1}{2} (q_0^2+q_1^2+q_2^2+q_3^2).
\end{align}
The action \eqref{sigmaaction} of the group $U(2)\rtimes H(2)$ on $\mathbb{T}$ expressed in $(q, \pi)$ assumes illegible form, so, we will not present it here. 
 
\section{Symplectic realizations of $\textbf{e}(3)^*$ defined by symplectic reduction}\label{sec:4}

In the previous section we have shown, see point (i) of Proposition \ref{cor:44}, that $\textbf{J}_e : (\widetilde{\mathbb{T}}, \Omega ) \to \mathbb{R}^2\times \dot{\mathbb{R}}^3$ is a full and complete symplectic realization of the Poisson submanifold $ \mathbb{R}^3 \times \dot{\mathbb{R}}^3$ of $ \textbf{e}(3)^*$. Here, applying to $(\widetilde{\mathbb{T}}, \Omega)$ symplectic reduction procedures, we obtain the three another symplectic realizations of $\textbf{e}(3)^*$. First one $(\widetilde{\mathbb{T}}\cap J_0^{-1}(\mu) )/U(1)  \to   \textbf{e}(3)^* $ is obtained by reduction of $(\widetilde{\mathbb{T}}, \Omega)$ to the levels of the  the map $J_0 : \widetilde{\mathbb{T}} \to \mathbb{R}$ defined in \eqref{428j}. The second one  $ (\widetilde{\mathbb{T}}\cap \Gamma_0^{-1}(\nu) )/\mathbb{R}  \to   \textbf{e}(3)^* $ is obtained by reduction of $(\widetilde{\mathbb{T}}, \Omega)$ to the levels of $\Gamma_0 : \widetilde{\mathbb{T}} \to \mathbb{R}$ defined in  \eqref{429j}.  The third one $(\widetilde{\mathbb{T}}\cap (J_0\times \Gamma_0)^{-1}(\mu , \nu ))/(U(1)\times \mathbb{R})\to  \textbf{e}(3)^*$ can be obtained using simultaneously both above reductions, i.e.  the reduction of $(\widetilde{\mathbb{T}}, \Omega)$ by $\textbf{J}_a=J_0\times \Gamma_0: \widetilde{\mathbb{T}} \to \mathbb{R}\times \mathbb{R}_-$. Let us begin from the third case.

\textbf{1. The case of $\textbf{J}_a=J_0\times \Gamma_0:\widetilde{\mathbb{T}}\to \mathbb{R}\times\mathbb{R}_-$.}

The degeneracy leaves of the restriction $\Omega|_{\textbf{J}^{-1}_a(\mu , \nu )}$ of the twistor symplectic form $\Omega$ to the submanifold
\begin{equation}\label{5a1}
\widetilde{\mathbb{T}}\cap (J_0\times \Gamma_0)^{-1}(\mu, \nu ) = \left\{\left(\begin{array}{c}
\vartheta \\
\zeta 
\end{array}\right) \in \widetilde{\mathbb{T}} : J_0(\vartheta , \zeta ) = \mu \mbox{ and } \Gamma_0 (\vartheta , \zeta ) = \nu\right\}
\end{equation}
are the orbits of the action \eqref{sigmaaction} of the subgroup $\tilde{A}(2)$ of $U(2)\ltimes H(2)$, see point (i) of Proposition \ref{prop:2}. From the points (i) and (iv) of Proposition \ref{cor:44} we conclude:

\begin{prop}\label{prop:51}
The reduced symplectic manifold $\widetilde{\mathbb{T}}\cap (J_0\times \Gamma_0)^{-1}(\mu, \nu )/ \tilde{A}(2)$ is symplectically isomorphic 
\begin{equation}
(\widetilde{\mathbb{T}}\cap (J_0\times \Gamma_0)^{-1}(\mu, \nu )/\tilde{A}(2) , \Omega_{\mu , \nu }) \cong (S_{\mu , \nu }, \omega_{\mu, \nu })
\end{equation}
with the $(\mu, \nu )$-level 
\begin{equation}\label{def53}
S_{\mu, \nu } := \textbf{J}^{-1}_{a,e} (\mu , \nu ) = \left\{(\vec{J}, \vec{\Gamma})\in \mathbb{R}^3 \times \dot{\mathbb{R}}^3 : \vec{J}\cdot \vec{\Gamma} = \mu \nu \mbox{ and } \vec{\Gamma}^2 = \nu^2 \right\}
\end{equation}
of $\textbf{J}_{a,e}: \mathbb{R}^3 \times \dot{\mathbb{R}}^3 \to \mathbb{R}\times \mathbb{R}_-$, which is an orbit of the coadjoint action \eqref{eaction} of $E(3)$ on $\textbf{e}(3)^*\cong \mathbb{R}^3\times \mathbb{R}^3$. The symplectic form $\Omega_{\mu , \nu }$ is the reduction of $\Omega$ to the quotient manifold $\widetilde{\mathbb{T}}\cap (J_0\times \Gamma_0)^{-1}(\mu, \nu )/\tilde{A}(2)$ and $\omega_{\mu , \nu }$ is symplectic form on the coadjoint orbit $S_{\mu, \nu }$ obtained by Kirillov construction. 
\end{prop}

Let us note that $-\nu >0$, so, $S_{\mu, \nu}$ is a $4$-dimensional $E(3)$-submanifold of $\mathbb{R}^3\times\dot{\mathbb{R}}^3 \subset \textbf{e}(3)^*$. The linear map 
\begin{equation}
R_\mu (\vec{J}, \vec{\Gamma}) := \left(\begin{array}{c}
\vec{J}-\mu \vec{\Gamma}\\
\vec{\Gamma}
\end{array}\right)
\end{equation}
of $\textbf{e}(3)^*$ into itself is equivariant with respect to the coadjoint action \eqref{eaction} of the group $E(3)$ and it defines a symplectic diffeomorphism
\begin{equation}
R_\mu : (S_{\mu , \nu }, \omega_{\mu, \nu }) \stackrel{\sim}{\rightarrow} (S_{0, \nu}, \omega_{0, \nu})
\end{equation}
between the coadjoint orbits, i.e. $\omega_{\mu , \nu} = R^*_\mu \omega_{0, \nu}$. From the definition \eqref{def53} one sees that $S_{0, \nu }\cong T\mathbb{S}^2_{-\nu} \cong T^*\mathbb{S}^2_{-\nu}$, where $\mathbb{S}^2_{-\nu}$ is a $2$-sphere of the radius $-\nu \in \mathbb{R}_+$. The bundle isomorphism  $T\mathbb{S}^2_{-\nu} \cong T^*\mathbb{S}^2_{-\nu}$ is done via Euclidean metric of $\mathbb{R}^3$. The symplectic form $\omega_{0, \nu}$ is the canonical symplectic form of $T^*\mathbb{S}^2_{-\nu}$. Explicit expression for $\omega_{\mu, \nu }= R^*_\mu \omega_{0, \nu}$ will be presented further.

Ending this subsection, let us note that $S_{\mu, \nu}$ is the total space of the bundle $\mbox{pr}_2: S_{\mu, \nu } \to \mathbb{S}^2_{-\nu}$. The bundle $\mbox{pr}_2:T^*\mathbb{S}^2_{-\nu}\cong S_{0, \nu } \to \mathbb{S}^2_{-\nu}$ is a vector bundle over $\mathbb{S}^2_{-\nu}$, which acts on the fibers of $\mbox{pr}_2: S_{\mu, \nu } \to \mathbb{S}^2_{-\nu}$, $\mu \neq 0$, making it an affine bundle over $\mathbb{S}^2_{-\nu}$. Hence, the map $R_\mu$ is identity covering morphism of these bundles.

\textbf{2. The case of $J_0 : \widetilde{\mathbb{T}} \to \mathbb{R}$.}

Let us consider $H_0(2)\times \dot{\mathbb{C}}^2$, where $  \dot{\mathbb{C}}^2:= \mathbb{C}^2 \backslash \{0\}$, as a real manifold with the left action $\Lambda : U(2) \ltimes H_0(2)\to \mbox{Diff}(H_0(2)\times  \dot{\mathbb{C}}^2)$ of the group $U(2)\ltimes H_0(2) $  on it defined by 
\begin{equation}\label{lambdag}
\Lambda_g \left(\begin{array}{c}
P \\
\zeta 
\end{array}\right) := \left(\begin{array}{c}
A(P+T)A^+ \\
A\zeta 
\end{array}\right),
\end{equation}
where $g=(A,T) \in U(2)\ltimes H_0(2)$ and $\left(\begin{array}{c}
P \\
\zeta 
\end{array}\right) \in H_0(2)\times \dot{\mathbb{C}}^2$. The orbits 
\begin{equation}\label{o57}
\mathcal{O}_\nu := H_0(2) \times \mathbb{S}^3_{\sqrt{-\nu}} = \{ (P, \zeta )\in H_0(2) \times \dot{\mathbb{C}}^2 : \zeta^+\zeta = -\nu \} \cong T^*\mathbb{S}^3_{\sqrt{-\nu}}
\end{equation}
of this action are parametrized by ${\sqrt{-\nu}} \in \mathbb{R}_+$, where 
\begin{equation}
\mathbb{S}_{\sqrt{-\nu}}^3 := \{\zeta \in \dot{\mathbb{C}}^2: \zeta^+\zeta = -\nu \}
\end{equation}
is the $3$-dimensional sphere in $\dot{\mathbb{C}}^2$ of radius $ \sqrt{-\nu}$.

Let us stress that the subgroup $\tilde{E}(3) = SU(2)\ltimes H_0(2) \subset U(2)\ltimes H_0(2)$ acts on $\mathcal{O}_\nu$ in a free and transitive way.

We also consider the map  $\Phi : H_0(2)\times \dot{\mathbb{C}}^2 \to \widetilde{\mathbb{T}}$ defined by 
\begin{equation}\label{eq:phi}
 H_0(2)\times \dot{\mathbb{C}}^2 \ni (P, \zeta ) \mapsto \Phi (P, \zeta ) := \left(\begin{array}{c}
\left(P-\frac{i\mu }{\zeta^+\zeta} \sigma_0\right)\zeta \\
\zeta 
\end{array}\right) \in \widetilde{\mathbb{T}} .
\end{equation}
One easily sees that $\Phi$ maps  $H_0(2)\times \dot{\mathbb{C}}^2$ in a smooth way on $ J_0^{-1}(\mu)\cap \widetilde{\mathbb{T}}$, where $J_0^{-1}(\mu)$ is the $\mu$-level set 
\begin{equation}
J^{-1}_0 (\mu) := \left\{\left(\begin{array}{c}
\vartheta \\
\zeta 
\end{array}\right) \in \mathbb{T}: J_0 (\vartheta , \zeta ) = \mu \right\}
\end{equation}
of the map $J_0: \mathbb{T} \to \mathbb{R}$, where $\mu \in \mathbb{R}$. 
\begin{prop}\label{prop:51}
\begin{itemize}
\item[\textbf{(i)}] The intersection $\widetilde{\mathbb{T}}\cap J^{-1}_0 (\mu)$ of $\widetilde{\mathbb{T}}$ with $J^{-1}_0 (\mu)$ is a  submanifold of $\widetilde{\mathbb{T}}$ and the map $\Phi : H_0(2)\times \dot{\mathbb{C}}^2 \to \widetilde{\mathbb{T}}\cap J^{-1}_0 (\mu) $ is a $U(2)\ltimes H_0(2)$-equivariant diffeomorphism, i.e. 
\begin{equation}\label{eq2}
\Phi \circ \Lambda_g = \Sigma_g \circ \Phi 
\end{equation}
for $g\in U(2)\ltimes H_0(2)$, where the actions $\Sigma_g$ and and $\Lambda_g$ are defined in \eqref{sigmaaction} and \eqref{lambdag}. In particular case, for $g=(e^{it} \sigma_0, 0)$, one has 
\begin{equation}\label{eq:54}
\Phi(P, e^{it}\zeta ) = e^{it} \Phi (P, \zeta ).
\end{equation}
\item[\textbf{(ii)}] The pull-back $\Phi^*\Theta$ of the one-form  
$\Theta $   defined in \eqref{thetaa} by $\Phi$ is given by 
\begin{equation}
\Phi^*\Theta = \zeta^+dP \zeta - 2 \frac{i\mu  }{\zeta^+\zeta} \zeta^+d\zeta + i\mu  d\log (\zeta^+\zeta ). 
\end{equation}
The one-forms $\Theta$ and $\Phi^*\Theta$ are invariant with respect to the actions $\Sigma_g$ and $\Lambda_g$, respectively. 
\item[\textbf{(iii)}] The superposition $\textbf{J}_e \circ \Phi : H_0(2) \times \dot{\mathbb{C}}^2 \to \mathbb{R}^3 \times \dot{\mathbb{R}}^3$ of $\Phi$ defined in \eqref{eq:phi} with the momentum map $\textbf{J}_e$ defined in \eqref{eq:je} is given by 
\begin{equation}\label{58}
(\textbf{J}_{e}\circ \Phi )  (P, \zeta ) = \left(\begin{array}{c}
\frac{i}{2} \zeta^+ [\vec{\sigma },P]\zeta + \frac{\mu}{\zeta^+\zeta} \zeta^+ \vec{\sigma}\zeta \\
- \zeta^+\vec{\sigma}\zeta 
\end{array}\right) 
\end{equation}
and it is a $U(2)\ltimes H_0(2)$-equivariant map, i.e.  $\textbf{J}_e\circ \Phi \circ \Lambda_g = \mbox{Ad}^*_{g^{-1}} \circ \Phi$, where for the definition of coadjoint action $\mbox{Ad}^*_{g^{-1}}$ and the action $\Lambda_g$  of $U(2)\ltimes H_0(2)$, see \eqref{adaction} and \eqref{eq:54}. 
\end{itemize}
\end{prop}
\begin{proof}
\textbf{(i)} The level set $J_0^{-1}(\mu)$ is invariant with respect to the action $\Sigma_g$ and the orbits $\mathcal{O}_{\mu , \nu}\subset J_0^{-1}(\mu)\cap \widetilde{\mathbb{T}}$ of this action, as well as the orbits $\mathcal{O}_\nu$, see \eqref{o57}, of the action $\Lambda_g$, are parametrized by $\nu= -\zeta^+\zeta\in \mathbb{R}_-$. Since in both cases the subgroup $SU(2)\ltimes H_0(2)$ acts in a free way on these orbits and $\Phi:H_0(2)\times \dot{\mathbb{C}}^2 \to J_0^{-1}(\mu)$ is a $SU(2)\times H_0(2)$-equivariant map, we obtain that it is a bijective smooth map. The equalities 
$$ \mbox{Im}(T\Phi(P, \zeta )) = \mbox{Ker}(TJ_0(\Phi(P, \zeta)))$$
and 
$$\mbox{dim}_\mathbb{R}\mbox{Ker}(TJ_0(\Phi(P, \zeta) )) = 7,$$
where $T\Phi(P, \zeta )$ and $TJ_0(\Phi(P, \zeta))$ are the derivatives of $\Phi$ at $(P, \zeta)$  and of $J_0$ at $\Phi(P, \zeta)$,
one shows by the straightforward calculations. Hence, $J_0^{-1}(\mu) \subset \widetilde{\mathbb{T}}$ is a submanifold of $\widetilde{\mathbb{T}}$ and $\Phi:H_0(2)\times \dot{\mathbb{C}}^2 \stackrel{\sim}{\rightarrow} J_0^{-1}(\mu)$ is a diffeomorphism. The $U(2)\ltimes H_0(2)$-equivariance  of $\Phi$ is proved by the straightforward verification. \\
Points $\textbf{(ii)}$ and $\textbf{(iii)}$ are established by straightforward verification also. 
\end{proof}

It follows from Proposition \ref{prop:51} that the subgroup $\{(A, T)  \in U(2)\ltimes H_0(2):A= e^{it}\sigma_0 , T=0  \} \cong U(1)$ acts on $H_0(2)\times \dot{\mathbb{C}}^2$ by \eqref{lambdag} and on $J_0^{-1}(\mu)$ by \eqref{sigmaaction} in a regular way. So, the quotients $(H_0(2)\times \dot{\mathbb{C}}^2)/U(1) \cong H_0(2) \times  \dot{\mathbb{C}}^2/U(1)$ and $J_0^{-1}(\mu)/U(1)$ of these manifolds by $U(1)$ are $6$-dimensional manifolds. Since of \eqref{eq:54} the map $\Phi$ defines the $\tilde{E}(3) $-equivariant diffeomorphism $\tilde{\Phi}: H_0(2) \times  \dot{\mathbb{C}}^2/U(1) \stackrel{\sim}{\rightarrow}J_0^{-1}(\mu)/U(1)$ of these quotient manifolds. We recall that $\tilde{E}(3)= SU(2)\ltimes H_0(2) $.

There are the coordinates $(\vec{p}, \vec{y})\in \mathbb{R}^3 \times \dot{\mathbb{R}}^3$ globally defined on $H_0(2) \times  \dot{\mathbb{C}}^2/U(1)$ by 
\begin{equation}
\vec{p} := -\frac{1}{2} \mbox{Tr}(\vec{\sigma}P), \qquad \vec{y} :=  \mbox{Tr}(\vec{\sigma}\zeta \zeta^+). 
\end{equation}
Therefore, we will identify $H_0(2) \times  \dot{\mathbb{C}}^2/U(1)$ with $\mathbb{R}^3 \times \dot{\mathbb{R}}^3$. The symplectic form $\Omega_\mu := \Phi^*d\Theta$ written in these coordinates  is the following 
\begin{equation}\label{510}
\Omega_\mu = d\vec{p} \wedge d\vec{y} - \frac{\mu}{||\vec{y}||^3} (y_1 dy_2\wedge dy_3 + y_2 dy_3 \wedge dy_1 + y_3 dy_1 \wedge dy_2).
\end{equation}
From the invariance of $\Theta$ with respect to the action of $\tilde{E}(3)= SU(2)\ltimes H_0(2)$ given in \eqref{sigmaaction}, it follows that $\Omega_\mu$ is invariant with respect to the standard action 
\begin{equation}
\tilde{\Lambda}_g \left(\begin{array}{c}
\vec{p}\\
\vec{y}
\end{array}\right) = \left(\begin{array}{c}
O(\vec{p} + \vec{\tau})\\
O\vec{y}
\end{array}\right),
\end{equation}
where $g= (O, \vec{\tau}) \in E(3) = SO(3)\ltimes \mathbb{R}^3$, of the Euclidean group $E(3)$ on $\mathbb{R}^3 \times \dot{\mathbb{R}}^3$. Let us stress that the symplectic form $\Omega_\mu $ is the sum of  the canonical symplectic form $d\vec{p} \wedge d\vec{y}$ of $T^*\dot{\mathbb{R}}^3 \cong \mathbb{R}^3 \times \dot{\mathbb{R}}^3$ and the $2$-form of a magnetic monopole of strenght $\mu$, e.g. see \cite{RM}.

The map $\textbf{J}_e\circ \Phi:H_0(2)\times \dot{\mathbb{C}}^2 \to \mathbb{R}^3\times \mathbb{R}^3\cong \textbf{e}(3)^*$, see \eqref{58}, is constant on the orbits of the  action of $U(1)$ on $H_0(2)\ltimes \dot{\mathbb{C}}^2$ defined in \eqref{lambdag}. So, it defines the map $\textbf{J}_{e, \mu}: H_0(2)\times \dot{\mathbb{C}}^2/ U(1) \cong \mathbb{R}^3 \times \dot{\mathbb{R}}^3  \hookrightarrow \mathbb{R}^3 \times \mathbb{R}^3 \cong \textbf{e}(3)^*$ of the quotient manifold $H_0(2)\times \dot{\mathbb{C}}^2/ U(1)$ into $\textbf{e}(3)^* $.

The function $\Gamma_0\circ \Phi : H_0(2)\times \dot{\mathbb{C}}^2 \to \mathbb{R}$ is invariant with respect to the action of $U(1)$ too. So, it defines a function  $\tilde{\Gamma}_0: H_0(2) \times  \dot{\mathbb{C}}^2/U(1)\cong \mathbb{R}^3 \times \dot{\mathbb{R}}^3 \to \mathbb{R}$ on the reduced phase space, which in coordinates $(\vec{p}, \vec{y})$ assumes the form 
\begin{equation}\label{tildeg0}
\tilde{\Gamma}_0 (\vec{p}, \vec{y}) = -||\vec{y}||. 
\end{equation}
The Hamiltonian flow on $\mathbb{R}^3 \times \dot{\mathbb{R}}^3$  generated by $\tilde{\Gamma}_0$ is the following 
\begin{equation}\label{action514}
\sigma_t (\vec{p}, \vec{y})  = \left(\begin{array}{c}
\vec{p} + t \frac{\vec{y}}{||y||}\\
\vec{y}
\end{array}\right).
\end{equation}

Summing up, we formulate:
\begin{prop}\label{prop:52}
\begin{itemize}
\item[\textbf{(i)}] The map $\tilde{\Phi} : H_0(2) \times  \dot{\mathbb{C}}^2/U(1) \stackrel{\sim}{\rightarrow}J_0^{-1}(\mu)/U(1)$ defines a $SU(2)\ltimes H_0(2)$-equivariant symplectic isomorphism between  $(H_0(2) \times  \dot{\mathbb{C}}^2/U(1) \cong \mathbb{R}^3 \times \dot{\mathbb{R}}^3 , \Omega_\mu )$ and $(J_0^{-1}(\mu)/U(1), d\tilde{\Theta}|_{J_0^{-1}(\mu)})$.
\item[\textbf{(ii)}] The map $\textbf{J}_{e, \mu}: (\mathbb{R}^3 \times \dot{\mathbb{R}}^3 , \Omega_\mu ) \to (\mathbb{R}^3 \times \dot{\mathbb{R}}^3 , \{\cdot , \cdot \}_{_{LP}})$ is a symplectic realization of the Lie-Poisson space $(\mathbb{R}^3 \times \dot{\mathbb{R}}^3 , \{\cdot , \cdot \}_{_{LP}})$ i.e. 
\begin{equation}
\{\textbf{J}_{e, \mu} \circ F, \textbf{J}_{e, \mu} \circ G\}_\mu = \{F, G\}_{_{LP}} \circ \textbf{J}_{e, \mu}, 
\end{equation}
for $F,G \in C^\infty (\mathbb{R}^3 \times \dot{\mathbb{R}}^3, \mathbb{R})$, 
 which in the coordinates $(\vec{p}, \vec{y})$ takes the form
\begin{equation}\label{jemi}
\textbf{J}_{e, \mu} \left(\vec{p},\vec{y}\right) = \left(\begin{array}{c}
 \vec{y}\times \vec{p} + \frac{\mu}{||\vec{y}||}\vec{y}\\
-\vec{y}
\end{array}\right) = \left(\begin{array}{c}
\vec{J}(\vec{p}, \vec{y})\\
\vec{\Gamma}(\vec{p}, \vec{y})
\end{array}\right) . 
\end{equation}
\item[\textbf{(iii)}] The image $\textbf{J}_{e, \mu} (\mathbb{R}^3\times \dot{\mathbb{R}}^3)$ of $\textbf{J}_{e, \mu }$ is a $5$-dimensional Poisson submanifold of $\textbf{e}(3)^*$ given by 
\begin{equation}
\textbf{J}_{e, \mu }(\mathbb{R}^3\times \dot{\mathbb{R}}^3) = \{(\vec{J}, \vec{\Gamma})\in \mathbb{R}^3\times \dot{\mathbb{R}}^3: \vec{J}\cdot \vec{\Gamma} = - \mu ||\vec{\Gamma}||\}
\end{equation}
and $\textbf{J}_{e, \mu}:  \mathbb{R}^3\times \dot{\mathbb{R}}^3\to \textbf{J}_{e, \mu} (\mathbb{R}^3\times \dot{\mathbb{R}}^3)$ is a surjective Poisson submersion on $\textbf{J}_{e, \mu} (\mathbb{R}^3\times \dot{\mathbb{R}}^3)$. The fibres $\textbf{J}_{e, \mu}^{-1} (\textbf{J}_{e, \mu} (\vec{p}, \vec{y}))$ of $\textbf{J}_{e, \mu}$ are the orbits of $(\mathbb{R}, +)$, which acts in a free and symplectic way on $(\mathbb{R}^3\times \dot{\mathbb{R}}^3, \Omega_\mu)$ by \eqref{action514}. 
Therefore, from the above and from Proposition 6.6 in \cite{W} we conclude that $\textbf{J}_{e, \mu}:  \mathbb{R}^3\times \dot{\mathbb{R}}^3\to \textbf{J}_{e, \mu} (\mathbb{R}^3\times \dot{\mathbb{R}}^3)$ is a full and complete symplectic realization of the Poisson manifold $\textbf{J}_{e, \mu} (\mathbb{R}^3\times \dot{\mathbb{R}}^3)$. 
\item[\textbf{(iv)}] The Poisson bracket $\{f, g \}_\mu$ of $f,g \in C^\infty (\mathbb{R}^3 \times \dot{\mathbb{R}}^3, \mathbb{R})$ defined by the symplectic form \eqref{510} is the following
\begin{equation}\label{pb15}
\{f, g \}_\mu = \frac{\partial f}{\partial \vec{p}}\frac{\partial g}{\partial \vec{y}}- \frac{\partial f}{\partial \vec{y}}\frac{\partial g}{\partial \vec{p}} -\frac{\mu}{||\vec{y}||^3} \vec{y}\cdot \left(\frac{\partial f}{\partial \vec{p}}\times \frac{\partial g}{\partial \vec{p}}\right)
\end{equation}
and from \eqref{pb15} one immediately obtains
\begin{equation}
\{y_k, y_l \}_\mu = 0, \quad \{p_k, y_l \}_\mu = \delta_{kl}, \quad \{p_k, p_l \}_\mu = -\frac{\mu}{||\vec{y}||^3}\epsilon_{klm}y_m .
\end{equation}
\end{itemize}
\end{prop}

Now let us reduce the symplectic form $\Omega_\mu$ given in \eqref{510} to the level $\tilde{\Gamma}_0^{-1}(\nu)$ of $\tilde{\Gamma}_0: \mathbb{R}^3 \times \dot{\mathbb{R}}^3 \to \mathbb{R}$. The leaves of degeneracy of $\Omega_\mu|_{\tilde{\Gamma}_0^{-1}(\nu)}$ are orbits of the action \eqref{action514} of $(\mathbb{R}, +)$. On the other hand side, they are levels of the map $\textbf{J}_{e, \mu}|_{\tilde{\Gamma}_0^{-1}(\nu)} : \tilde{\Gamma}_0^{-1}(\nu)\to \textbf{e}(3)^*\cong \mathbb{R}^3\times \mathbb{R}^3$, where $\textbf{J}_{e, \mu }$ is given in \eqref{jemi}. Taking into account the above facts we obtain the $E(3)$-equivariant symplectic diffeomorphism $\textbf{J}_{e, \mu, \nu}: \tilde{\Gamma}_0^{-1}(\nu)/\mathbb{R} \stackrel{\sim}{\rightarrow}S_{\mu, \nu}$ of the reduced symplectic manifold $(\tilde{\Gamma}_0^{-1}(\nu)/\mathbb{R} , \tilde{\Omega}_{\mu, \nu})$ with the coadjoint orbit (symplectic leaf) $(S_{\mu, \nu}, \omega_{\mu, \nu})$, see \eqref{def53} in Proposition \ref{prop:51}.

Therefore, using $\textbf{J}_{e, \mu, \nu}: \tilde{\Gamma}_0^{-1}(\nu)/\mathbb{R} \stackrel{\sim}{\rightarrow}S_{\mu, \nu}$ we can obtain the explicit form of $\omega_{\mu, \nu}$. For this reason we note that arbitrary orbit of the action \eqref{action514} intersects the submanifold 
\begin{equation}\label{defmmini}
M_{\mu, \nu}:= \left\{\left(\begin{array}{c}
\vec{p}\\
\vec{y}
\end{array}\right) \in \tilde{\Gamma}_0^{-1}(\nu): \vec{p}\cdot \vec{y} =0\mbox{ and } \vec{y}^2 = \nu^2\right\}
\end{equation}
in one point only. So, we can identify $\tilde{\Gamma}_0^{-1}(\nu)/\mathbb{R}$ with $M_{\mu, \nu}$ and the reduced symplectic form $\tilde{\Omega}_{\mu, \nu}$ is obtained by the restriction $\Omega_\mu|_{M_{\mu, \nu }}$ of $\Omega_\mu$ to $M_{\mu, \nu}$. We see from the definition \eqref{defmmini} that $M_{\mu, \nu} \cong T\mathbb{S}^2_{-\nu} \cong T^*\mathbb{S}^2_{-\nu}$. It follows from Proposition \ref{prop:ap} (see Appendix), that the reduction of the canonical part $d\vec{y}\wedge d\vec{p}$ of $\Omega_\mu$ to $M_{\mu, \nu}\cong T^*\mathbb{S}^2_{-\nu}$ is the canonical symplectic form $d\gamma_\nu$ of the cotangent bundle $T^*\mathbb{S}^2_{-\nu}$,  where $\gamma_\mu$ is the Liouville form of $T^*\mathbb{S}^2_{-\nu}$. The magnetic monopole part of $\Omega_\mu$, after reduction to $T^*\mathbb{S}^2_{-\nu}$, gives the pullback $\pi^*\sigma_{-\nu}$ by the bundle projection map $\pi: T^*\mathbb{S}^2_{-\nu} \to \mathbb{S}^2_{-\nu}$ of the surface $2$-form $\sigma_{-\nu}$ of $\mathbb{S}^2_{-\nu}$. Hence, we obtain
\begin{equation}\label{omini}
\tilde{\Omega}_{\mu , \nu} = d\gamma_\nu - \frac{\mu}{\nu^2} \sigma, 
\end{equation}
where $\sigma$ is the surface form of the $2$-sphere $\mathbb{S}^2$ of radius one.

Ending this subsection we mention that the Poisson map \eqref{jemi} restricted to $M_{\mu, \nu}\subset \mathbb{R}^3 \times \dot{\mathbb{R}}^3$ defines a symplectic diffeomorphism 
\begin{equation}
\textbf{J}_{e, \mu , \nu} : (M_{\mu, \nu }, \tilde{\Omega}_{\mu, \nu}) \stackrel{\sim}{\rightarrow} (S_{\mu, \nu}, \omega_{\mu, \nu})
\end{equation}
of the reduced symplectic manifold with the suitable coadjoint orbit of $E(3)$. The equalities $\textbf{J}_{e, \mu, \nu }^* \omega_{\mu, \nu} = \tilde{\Omega}_{\mu, \nu}$ and \eqref{omini} explain geometric structure of the Kirillov symplectic form $\omega_{\mu, \nu}$. Let us also mention that the reduction by $J_0: \widetilde{\mathbb{T}} \to \mathbb{R}$ and next through $\tilde{\Gamma}_0: \mathbb{R}^3\times \dot{\mathbb{R}}^3 \to \mathbb{R}$ gives the same result as the reduction through $J_0\times \Gamma_0: \widetilde{\mathbb{T}}\to \mathbb{R}\times \mathbb{R}_-$.

\textbf{3. The case of $\Gamma_0 :\widetilde{\mathbb{T}}\to \mathbb{R}$.}

In this subsection we describe the reduction of the twistor symplectic structure $(\widetilde{\mathbb{T}}, \Omega)$  to the level sets
\begin{equation}
\Gamma_0^{-1} \left(\nu\right) := \left\{ \left(\begin{array}{c}
\vartheta \\
\zeta 
\end{array}\right) \in \widetilde{\mathbb{T}} : \zeta^+\zeta = -\nu \right\},
\end{equation}
where $\nu \in \mathbb{R}_-$, of the function $\Gamma_0:\widetilde{\mathbb{T}}\to \mathbb{R}_-$ defined in \eqref{429j}. The Hamiltonian flow $\sigma_t : \widetilde{\mathbb{T}}\to \widetilde{\mathbb{T}}$ generated by $\Gamma_0$ is given by 
\begin{equation}\label{act5}
\sigma_t \left(\begin{array}{c}
\vartheta \\
\zeta 
\end{array}\right) = \left(\begin{array}{c}
\vartheta + t \zeta\\
\zeta 
\end{array}\right)
\end{equation}
and it preserves the level submanifold $\Gamma_0^{-1} \left(\nu\right) \cong \mathbb{C}^2 \times \mathbb{S}_{\sqrt{-\nu}}^3$. The action \eqref{act5} of $(\mathbb{R}, +)$ on $\Gamma_0^{-1}(\nu)$ is regular, so, the space $\Gamma_0^{-1}(\nu)/\mathbb{R}$ of $(\mathbb{R}, +)$-orbits is a $6$-dimensional symplectic manifold, which symplectic form $\Omega_\nu$ is obtained by the reduction of the twistor symplectic form $\Omega$ to the quotient $\Gamma_0^{-1}(\nu)/\mathbb{R}$.

Let us now describe the structure of $(\Gamma_0^{-1}(\nu)/\mathbb{R}, \Omega_{\nu})$ in details. We note that $\Gamma_0^{-1}(\nu)/\mathbb{R}$ is diffeomorphic $\Gamma_0^{-1}(\nu)/\mathbb{R} \cong S_{\nu }$ to the submanifold
\begin{equation}\label{smi}
S_{\nu  } := \left\{\left(\begin{array}{c}
\vartheta \\
\zeta \end{array}\right)\in \widetilde{\mathbb{T}} : \zeta^+\zeta = -\nu \mbox{ and } \zeta^+\vartheta + \vartheta^+ \zeta = 0 \right\}
\end{equation}
of $\widetilde{\mathbb{T}}$. The above statement follows from the observation that the orbits of the flow \eqref{act5} intersect $S_{\nu  }\subset \Gamma_0^{-1}(\nu)$ in the one point only.

The action $\Sigma_{(A, T)}: \widetilde{\mathbb{T}}\to \widetilde{\mathbb{T}}$ of $(A,T)\in \tilde{E}(3)$ on $\widetilde{\mathbb{T}}$ defined in \eqref{sigmaaction} preserves $\Gamma_0^{-1}(\nu)$ and commutes with the action \eqref{act5} of $(\mathbb{R}, +)$. So, it defines the symplectic action $\widetilde{\Sigma}_{(A,T)}: \Gamma_0^{-1}(\nu)/\mathbb{R} \to \Gamma_0^{-1}(\nu)/\mathbb{R}$ of $\tilde{E}(3)$ on the reduced symplectic manifold $(\Gamma_0^{-1}(\nu)/\mathbb{R}, \Omega_{\nu})$. The diffeomorphism $\Gamma_0^{-1}(\nu)/\mathbb{R}\cong S_{\nu}$ allows us to transport this action on the action $\Sigma_{\nu_{(A,T)}}: S_\nu \to S_\nu$ of $\tilde{E}(3)$ on $S_\nu$, which is given by 
\begin{equation}
S_\nu \ni \left(\begin{array}{c}
\vartheta\\
\zeta \end{array}\right) \mapsto \Sigma_{\nu_{(A,T)}}\left(\begin{array}{c}
\vartheta\\
\zeta \end{array}\right) = \left(\begin{array}{c}
A\left(\vartheta +\left(T -\frac{\zeta^+T\zeta}{\zeta^+\zeta}\right)\zeta\right)\\
A\zeta \end{array}\right) \in S_\nu . 
\end{equation}

Let us now define the diffeomorphism $\Psi:\mathbb{S}^3_{\sqrt{-\nu}}\times iH_0(2)\stackrel{\sim}{\rightarrow} S_\nu$ as follows 
\begin{equation}
\mathbb{S}^3_{\sqrt{-\nu}}\times iH_0(2) \ni (\zeta , P) \mapsto \Psi (\zeta, P) := \left(\begin{array}{c}
P\zeta \\
\zeta
\end{array}\right) \in S_\nu . 
\end{equation}
The diffeomorphism $\Psi^{-1} : S_\nu \stackrel{\sim}{\rightarrow} \mathbb{S}^3_{\sqrt{-\nu}}\times iH_0(2) $ inverse to $\Psi: \mathbb{S}^3_{\sqrt{-\nu}}\times iH_0(2)\stackrel{\sim}{\rightarrow} S_\nu$ is given by
\begin{equation}
S_\nu \ni (\vartheta , \zeta ) \mapsto \Psi^{-1} (\vartheta , \zeta ) = (\zeta , [\vartheta , i\sigma_2 \bar \vartheta ][\zeta , i\sigma_2 \bar\zeta]^{-1}) \in \mathbb{S}^3_{\sqrt{-\nu}}\times iH_0(2),
\end{equation}
where $[\vartheta , i\sigma_2 \bar \vartheta ] = \left[\left(\begin{array}{c}
\vartheta_1\\
\vartheta_2
\end{array}\right), i\sigma_2 \left(\begin{array}{c}
\bar\vartheta_1\\
\bar\vartheta_2
\end{array}\right)\right]\in \mbox{Mat}_{2\times 2}(\mathbb{C})$ and the $2\times 2$ matrix $[\zeta , i\sigma_2 \bar\zeta]$ is defined in the same way. The action $\Lambda_\nu (A,T):= \Psi^{-1}\circ \Sigma_{\nu_{(A,T)}}\circ \Psi : \mathbb{S}^3_{\sqrt{-\nu}}\times iH_0(2)\to \mathbb{S}^3_{\sqrt{-\nu}}\times iH_0(2)$ of $\tilde{E}(3)$ on $\mathbb{S}^3_{\sqrt{-\nu}}\times iH_0(2)$ assumes the form 
\begin{equation}\label{ac532}
\Lambda_\nu (A,T) (\zeta , P) = \left(A\zeta , A\left(P+\frac{1}{\zeta^+\zeta}[T, \zeta\zeta^+]\right)A^+\right). 
\end{equation}
The pullback $\Psi^*\Omega_\nu$ of the reduced symplectic form $\Omega_\nu$ on $\mathbb{S}^3_{\sqrt{-\nu}}\times iH_0(2)$ is given by 
\begin{equation}
\Psi^*\Omega_\nu = d(\zeta^+Pd\zeta - d\zeta^+P\zeta) = d\mbox{Tr}((d\zeta\zeta^+ - \zeta d\zeta^+)P)
\end{equation}
and it is invariant with respect to the action \eqref{ac532}. Let us recall here that $\zeta^+\zeta = -\nu = const, $ i.e. $\zeta \in \mathbb{S}^3_{\sqrt{-\nu}}$.

The function $J_0$ is invariant with respect to \eqref{act5} and  it defines a function on the reduced phase space $\tilde{J}_0 : \mathbb{S}^3_{\sqrt{-\nu}}\times iH_0(2)\cong S_\nu \to \mathbb{R}$. In coordinates $(\zeta, P)$ it assumes the form 
\begin{equation}\label{tildej0}
\tilde{J}_0 (\zeta , P) = -i\zeta^+P\zeta. 
\end{equation}
The Hamiltonian flow on $\mathbb{S}^3_{\sqrt{-\nu}}\times iH_0(2)\cong S_\nu$  generated by $\tilde{J}_0$ is the following 
\begin{equation}\label{action524}
\sigma_t (\zeta , P)  = (e^{it} \zeta , P). 
\end{equation}

The equations on $S_\nu$, see \eqref{smi}, written in the real coordinates $(q, \pi ) \in \mathbb{R}^4\times \mathbb{R}^4$, defined by \eqref{rcn}, assumes the following form
\begin{equation}\label{eq:529}
\begin{array}{l}
q_0\pi_0+q_1\pi_1 + q_2\pi_2 + q_3\pi_3=0, \\
q_0^2 +q_1^2 +q_2^2+q_3^2 = -2\nu ,
\end{array}
\end{equation}
i.e. they are equations on the tangent bundle $T\mathbb{S}^3_{\rho} \subset T\mathbb{R}^4 \cong \mathbb{R}^4\times \mathbb{R}^4$ of the $3$-dimensional sphere $\mathbb{S}^3_{\rho}$ of radius $\rho :=\sqrt{-2\nu}$. Further we will identify the cotangent bundle $T^*\mathbb{S}^3_{\rho}$ with the tangent one $T\mathbb{S}^3_{\rho}$ through the Euclidean metric tensor on $\mathbb{R}^4$.

Let us mention also that $\mathbb{S}^3_{\rho}$ is diffeomorphic $\mathbb{S}^3_{\rho} \cong SU(2)$ with the special unitary group $SU(2)$. This implies the triviality   of the tangent $T\mathbb{S}^3_{\rho} \cong \mathbb{S}^3_{\rho} \times \mathbb{R}^3$ and cotangent $T^*\mathbb{S}^3_{\rho} \cong \mathbb{S}^3_{\rho} \times (\mathbb{R}^3)^*$ bundles of $\mathbb{S}^3_{\rho}$.

Since $\mathbb{S}^3_{\rho}$ is a submanifold $\iota_\nu : \mathbb{S}^3_{\rho} \hookrightarrow\mathbb{R}^4$ of $\mathbb{R}^4$, so, the reduction of $\Omega$, defined in \eqref{2formr}, to $\iota^*_\nu T^*\mathbb{R}^4 := \{ (q; \varphi ) \in \mathbb{S}^3_{\rho} \times T^*\mathbb{R}^4 : \iota_\nu (x,y) = \pi^*(\varphi)\}$ defines the symplectic form $\tilde{\Omega}_\nu$ on the space $\iota^*_\nu T^*\mathbb{R}^4/\mathbb{R}$ of degeneracy leaves of $(\iota^*_\nu)^*\Omega$ isomorphic $(\iota^*_\nu T^*\mathbb{R}^4/\mathbb{R} ) \cong (\Gamma_0^{-1}\left(\nu\right)/\mathbb{R})$ to the space of orbits of the action \eqref{act5} of $(\mathbb{R}, +)$. From the above and the general statement (see Proposition \ref{prop:ap} in the Appendix) we have 
\begin{prop}
\begin{itemize}
\item[\textbf{(i)}] One has the following $(S_\nu , \Omega_\nu)$,  $(\mathbb{S}^3_{\sqrt{-\nu}} \times iH_0(2), \Psi^*\Omega_\nu)$ and $(T^*\mathbb{S}^3_{\rho}, d\gamma_\rho)$ symplectically isomorphic realizations of  the reduced $\tilde{E}(3)$-symplectic manifold  $\left(\Gamma_0^{-1}\left(\nu\right)/\mathbb{R}, \Omega_\nu \right)$, where $d\gamma_\rho$ is canonical symplectic form on $T^*\mathbb{S}^3_{\rho}$. 
\item[\textbf{(ii)}] The momentum map $\textbf{J}_e: \widetilde{\mathbb{T}}\to \textbf{e}(3)^*$ restricted to $S_{\nu  }$ defines surjective submersion $\textbf{J}_{e, \nu }: S_{\nu  } \to \mathbb{R}^3 \times \mathbb{S}^2_{-\nu}$ of $S_\nu$ on the submanifold  $\mathbb{R}^3 \times \mathbb{S}^2_{-\nu}$ of $ \textbf{e}(3)^*\cong \mathbb{R}^3 \times \mathbb{R}^3 $, where $\mathbb{S}^2_{-\nu} := \{ \vec{\Gamma}\in \dot{\mathbb{R}}^3: \vec{\Gamma}^2 = \nu^2 \}$ is a $2$-sphere in $\dot{\mathbb{R}}^3$ of radius $-\nu$. The fibres $\textbf{J}_{e, \nu }^{-1} (\textbf{J}_{e, \nu }(\vartheta, \zeta ))$ of $\textbf{J}_{e, \nu }: S_\nu \to \textbf{J}_{e, \nu }(S_\nu)\cong \mathbb{R}^3 \times \mathbb{S}^2_{-\nu}$ are the orbits of $U(1)$ which acts on $S_\nu $ by \eqref{action524} in a free and symplectic way. Therefore, $\textbf{J}_{e, \nu }: S_\nu \to \textbf{J}_{e, \nu }(S_\nu)\cong \mathbb{R}^3 \times \mathbb{S}^2_{-\nu}$ is a full and complete (see Proposition 6.6. in \cite{W}) symplectic realization of the Poisson submanifold $\mathbb{R}^3 \times \mathbb{S}^2_{-\nu} \subset \textbf{e}(3)^*$.
\end{itemize}
\end{prop}

In order to obtain a canonical coordinate description of the symplectic realization $\textbf{J}_{e, \nu}: T^*\mathbb{S}^3_\rho\to \textbf{e}(3)^*$, we take the stereographic coordinates 
\begin{equation}\label{537}
\mathbb{S}^3_\rho \backslash \left\{\left(\begin{array}{c}
\rho\\
 \vec{0}
\end{array}\right)\right\} \ni \left(\begin{array}{c}
q_0\\
\vec{q}
\end{array}\right) \stackrel{\sim}{\rightarrow} \vec{y} := \frac{1}{\rho - q_0}\vec{q} \in \mathbb{R}^3  
\end{equation}
on the $3$-sphere $\mathbb{S}^3_\rho $ with removed point $(\rho, \vec{0})\in \mathbb{S}^3_\rho $. The inverse of \eqref{537} is given by 
\begin{equation}\label{538i}
\left(\begin{array}{c}
q_0\\
\vec{q}
\end{array}\right) = \frac{\rho}{1+\vec{y}^2} \left(\begin{array}{c}
\vec{y}^2-1\\
2\vec{y}
\end{array}\right). 
\end{equation}
One obtains the Moser canonical coordinates $(\vec{y}, \vec{p})\in \mathbb{R}^3\times \mathbb{R}^3$ on $T^*\left(\mathbb{S}^3_\rho \backslash \left\{\left(\begin{array}{c}
\rho\\
 \vec{0}
\end{array}\right)\right\}\right)$  from equality 
\begin{equation}\label{539}
\pi_0dq_0 + \vec{\pi}d\vec{q} = \vec{p}d\vec{y}, 
\end{equation}
e.g. see \cite{Pe, Mos}. Taking the differentials of \eqref{538i} and substituting them into \eqref{539}  one finds that 
\begin{equation}\label{eq:542}
\vec{p} = \rho \left(\frac{4\vec{y}}{(\vec{y}^2+1)^2} (\pi_0 - \vec{\pi}\cdot \vec{y}) + \frac{2}{\vec{y}^2 +1} \vec{\pi}\right). 
\end{equation}
From the above equality one obtains 
\begin{equation}
\pi_0 = \frac{1}{\rho} \vec{y}\cdot \vec{p} \quad \mbox{ and } \quad \vec{\pi}\cdot \vec{y} = \frac{1}{\rho}\frac{(1+\vec{y}^2)^2}{2(1-\vec{y}^2)} \vec{y}\cdot \vec{p}, 
\end{equation}
and, thus
\begin{equation}\label{541}
\left(\begin{array}{c}
\pi_0 \\
\vec{\pi} 
\end{array}\right) = \frac{1}{\rho} \left(\begin{array}{c}
\vec{y}\cdot \vec{p}\\
\frac{\vec{y}^2+1}{2}\vec{p} - (\vec{y}\cdot \vec{p})\vec{y}
\end{array}\right). 
\end{equation}
Hence, in the coordinates $(\vec{y}, \vec{p})$, defined by \eqref{537} and \eqref{541}, the symplectic form $\Omega_\nu$ and the momentum map $\textbf{J}_{e, \nu}: T^*\left(\mathbb{S}^3_\rho \backslash \left\{\left(\begin{array}{c}
\rho\\
 \vec{0}
\end{array}\right)\right\}\right)\to \textbf{e}(3)^*$ assume the forms 
\begin{equation}
\Omega_\nu = d\vec{p} \wedge d\vec{q}, 
\end{equation}
and 
\begin{equation}\label{eq:546}
\textbf{J}_{e, \nu } (\vec{y}, \vec{p}) = \left(\begin{array}{c}
\vec{J}(\vec{y}, \vec{p})\\
\vec{\Gamma} (\vec{y}, \vec{p})
\end{array}\right) ,
\end{equation}
where 
\begin{equation}\label{eq:547}
\begin{array}{l}
\vec{J} (\vec{y}, \vec{p})= \frac{1}{2}\left(\begin{array}{c}
\frac{\vec{y}^2-1}{2} p_3 - (\vec{y}\cdot \vec{p}) y_3 + y_1p_2-y_2p_1\\
\frac{\vec{y}^2-1}{2} p_1 - (\vec{y}\cdot \vec{p}) y_1 + y_2p_3-y_3p_2\\
\frac{\vec{y}^2-1}{2} p_2 - (\vec{y}\cdot \vec{p}) y_2 + y_3p_1-y_1p_3
\end{array}\right),\\
\vec{\Gamma} (\vec{y}, \vec{p})= \frac{2\rho^2}{(1+\vec{y}^2)^2}\left(\begin{array}{c}
(\vec{y}^2-1)y_1 + 2y_2y_3\\
-(\vec{y}^2-1)y_3 + 2y_1y_2\\
\frac{1}{4} (\vec{y}^2-1)^2 +y_1^2-y_2^2-y_3^2
\end{array}\right). 
\end{array}
\end{equation}

Summing up, we illustrate the symplectic realizations of $\textbf{e}(3)^*$ obtained in this and previous section as well as symplectic reductions used for their constructions in the following diagram 

\begin{equation}\label{538}
\begin{tikzcd}
                                                                                                    & {(\widetilde{\mathbb{T}}, \Omega)} \arrow[dashed]{ld}[swap]{\Gamma_0} \arrow[dashed]{rd}{J_0} \arrow{d}{\textbf{J}_e} &                                                                                                                            \\
{(T^*\mathbb{S}^3_\rho , \Omega_\nu )} \arrow[dashed]{rd}[swap]{\tilde{J}_0} \arrow{r}{\textbf{J}_{e, \nu}} & \textbf{e}(3)^*                                                                                         & {(\mathbb{R}^3\times \dot{\mathbb{R}}^3 , \Omega_\mu )} \arrow[dashed]{ld}{\tilde{\Gamma}_0}\arrow{l}[swap]{\textbf{J}_{e , \mu}} \\
                                                                                                    & {(M_{\mu , \nu }, \Omega_{\mu, \nu})} \arrow[hook]{u}[swap]{\textbf{J}_{e, \mu , \nu }}                                                   &                                                                                                                           
\end{tikzcd},
\end{equation}
where by dashed arrows we denoted not the maps but the corresponding to them reduction procedures.

The realization $\textbf{J}_{e, \mu , \nu }: M_{\mu, \nu } \to \textbf{e}(3)^*$ maps $M_{\mu, \nu } $ on a symplectic leaf $S_{\mu, \nu }$ of $\textbf{e}(3)^*$  defined in \eqref{def53}. So, it reduces  the gyrostat with a fixed point system to a $4$-dimensional symplectic manifold defined by the Casimir functions of Lie-Poisson space $\textbf{e}(3)^*$. Therefore, it is a natural and standard step in integration procedure of the system.

The others  Poisson maps $\textbf{J}_e: (\widetilde{\mathbb{T}}, \Omega) \to (\mathbb{R}^3\times \dot{\mathbb{R}}^3, \{\cdot, \cdot \}_{_{LP}})$, $\textbf{J}_{e, \mu}: (\mathbb{R}^3\times \dot{\mathbb{R}}^3, \Omega_\mu) \to (\textbf{J}_{e, \mu}(\mathbb{R}^3\times \dot{\mathbb{R}}^3), \{\cdot, \cdot \}_{_{LP}})$ and $\textbf{J}_{e, \nu}: (T^*\mathbb{S}^3_\rho, \Omega_\nu) \to (\mathbb{R}^3\times \mathbb{S}^2_{-\nu}, \{\cdot, \cdot \}_{_{LP}})$ from the diagram \eqref{538} are full and complete symplectic realizations of the corresponding Poisson submanifolds of $\textbf{e}(3)^*$. We believe that they are unknown and lead to large family of the new integrable Hamiltonian systems, defined by  the integrable cases of gyrostat as well as heavy top systems. We will discuss some of them in the next section.

\section{Integrable Hamiltonian systems on the symplectic realizations of $\textbf{e}(3)^*$}\label{sec:5}

Let us recall that on the Lie-Poisson space $\textbf{e}(3)^*$ there are two functionally independent Casimir functions $K_1, K_2 \in C^\infty(\textbf{e}(3)^*, \mathbb{R})$, see \eqref{casimirs}, what implies that the symplectic leaves of $\textbf{e}(3)^*$ can have dimension not more than $4$. So, for the integrability of a Hamiltonian system on $\textbf{e}(3)^*$, except its Hamiltonian $H\in C^\infty (\textbf{e}(3)^*, \mathbb{R})$, one needs only one additional integral of motion $K\in C^\infty(\textbf{e}(3)^*, \mathbb{R})$.  Therefore, having done $K$ we obtain the systems of functionally independent integrals of motion in involution: $(H\circ \textbf{J}_{e, \mu}, K_2 \circ \textbf{J}_{e, \mu}, K\circ \textbf{J}_{e, \mu} )$, $(H\circ \textbf{J}_{e, \nu}, K_1\circ \textbf{J}_{e, \nu}, K\circ \textbf{J}_{e, \nu})$ and $(H\circ \textbf{J}_{e}, K_1 \circ \textbf{J}_{e}, K_2\circ \textbf{J}_{e}, K\circ \textbf{J}_{e})$ on the symplectic manifolds $\mathbb{R}^3\times \dot{\mathbb{R}}^3$, $T^*\mathbb{S}^3_\rho$ and $\widetilde{\mathbb{T}}$, respectively. The above makes integrable  the systems defined by the Hamiltonians $H\circ \textbf{J}_{e, \mu}\in C^\infty (\mathbb{R}^3\times \dot{\mathbb{R}}^3, \mathbb{R})$, $H\circ \textbf{J}_{e,\nu}\in C^\infty (T^*\mathbb{S}^3_\rho, \mathbb{R})$ and $H\circ \textbf{J}_{e}\in C^\infty (\widetilde{\mathbb{T}}, \mathbb{R})$ being the lifting of the Hamiltonian $H_{\lambda} \in C^\infty (\textbf{e}(3)^*, \mathbb{R})$ defined in \eqref{Hht}. These Hamiltonians are examples of the ones which describe so-called collective motion considered in \cite{G1,G2,MF}.

In this section we discuss some of these systems in details indicating their physical applications too.

\subsection{Integrable Hamiltonian systems on $(\widetilde{\mathbb{T}}, \Omega)$}

The twistor symplectic form \eqref{2form} on $\mathbb{T}$ written in spinor coordinates  $a,b \in \mathbb{C}^2 $ related by 
\begin{equation}\label{eq:61}
\left(\begin{array}{c}
\vartheta \\
\zeta 
\end{array}\right)= \frac{1}{\sqrt{2}}\left(\begin{array}{c}
a+i\bar b\\
\bar b + i a
\end{array}\right) 
\end{equation}
to the spinor coordinates $\vartheta , \zeta \in \mathbb{C}^2$ defined in \eqref{eq:41}, assumes the canonical form 
\begin{equation}
\Omega = -i (da^+\wedge da + db^+\wedge db ), 
\end{equation}
where $\bar b =\left(\begin{array}{c}
\bar b_1\\
\bar b_2
\end{array}\right)$ and $b^T= (b_1, b_2)$ are complex conjugation and transposition of $b= \left(\begin{array}{c}
b_1\\
b_2
\end{array}\right)\in \mathbb{C}^2$. Substituting \eqref{eq:61} into \eqref{428} and \eqref{429} one easily finds that the vector coordinates $(\vec{J}, \vec{\Gamma})$ of the momentum map $\textbf{J}_e: \widetilde{\mathbb{T}}\to \textbf{e}(3)^*\cong \mathbb{R}^3\times \mathbb{R}^3$ depend on $(a,b)$ by 
\begin{align}\label{eq:63}
\vec{J} & = \frac{1}{2} (a^+\vec{\sigma} a - b^T \vec{\sigma} \bar b),\\
\label{eq:64}
\vec{\Gamma} & = \frac{1}{2} (b^T \vec{\sigma } \bar b + a^+ \vec{\sigma } a - i a^+ \vec{\sigma} \bar b + i b^T \vec{\sigma } a).
\end{align}
Substituting $\vec{J}$ and $\vec{\Gamma }$ given by \eqref{eq:63} and \eqref{eq:64} into \eqref{Hht} one finds that  
\begin{multline}\label{eq:65}
H_{\lambda}\circ \textbf{J}_e= \frac{1}{8} \left(\frac{1}{I_1}-\frac{1}{I_2}\right) [a_1^2\bar a_2^2 + \bar a_1^2 a_2^2 +b_1^2\bar b_2^2 + \bar b_1^2 b_2^2 - 2\bar a_1 a_2 b_1 \bar b_2 - 2 a_1 \bar a_2 \bar b_1 b_2 ] +\\
\frac{1}{8} \left(\frac{1}{I_1}+\frac{1}{I_2}\right) [2|a_1|^2 |a_2|^2 + 2 |b_1|^2|b_2|^2 - 2 \bar a_1 a_2 \bar b_1 b_2 - 2 a_1 \bar a_2b_1 \bar b_2]+\\
\frac{1}{8I_3} [|a_1|^2 -|a_2|^2 -|b_1|^2 + |b_2|^2]^2 + \frac{\lambda_1}{2I_1}( \bar a_1 a_2 + a_1 \bar a_2 -\bar b_1 b_2 - b_1 \bar b_2 ) + \\
\frac{i\lambda_2}{2I_2}( \bar a_1 a_2 - a_1 \bar a_2 +\bar b_1 b_2 - b_1 \bar b_2 ) + \frac{\lambda_3}{2I_3}( |a_1|^2 -|a_2|^2 -|b_1|^2 + |b_2|^2 ) + \frac{\lambda_1^2}{2I_1} + \frac{\lambda_2^2}{2I_2}+ \frac{\lambda_3^2}{2I_3}+ (U\circ\vec{\Gamma})(a, b),
\end{multline}
where $\vec{\Gamma } : \widetilde{\mathbb{T}}\to \dot{\mathbb{R}}^3$ is given by \eqref{eq:64}. We recall that Hamiltonian $H_{\lambda}$ is defined in \eqref{Hht}. The integrals of motion $K_1\circ \textbf{J}_e$ and $K_2 \circ \textbf{J}_e$ one finds from the dependence 
\begin{equation}\label{eq:66}
\left(\begin{array}{c}
K_1 \\
K_2 
\end{array}\right) = \left(\begin{array}{c}
\vec{J}\cdot \vec{\Gamma }  \\
\vec{\Gamma }^2
\end{array}\right)= \left(\begin{array}{c}
J_0 \Gamma_0 \\
\Gamma_0^2
\end{array}\right),
\end{equation}
for which see \eqref{casfun} and \eqref{casfun2}, where 
\begin{align}\label{eq:67}
J_0 & = -\frac{1}{2}(a^+a - b^+b)= \frac{1}{2}(|b_1|^2 +|b_2|^2 - |a_1|^2 -|a_2|^2),\\
\nonumber
\Gamma_0 &= -\frac{1}{2} (a^+a + b^+b +i (b^Ta - a^+ \bar b))= \\
\label{eq:68}
 & - \frac{1}{2} (|a_1|^2 + |a_2|^2 + |b_1|^2 + |b_2|^2 + i (b_1a_1+b_2a_2 -\bar b_1\bar a_1 - \bar b_2 \bar a_2)). 
\end{align}
One obtains equalities \eqref{eq:67} and \eqref{eq:68} from \eqref{428j} and \eqref{429j} by using \eqref{eq:61}.

Taking into account dependence \eqref{eq:66}, we see that instead of $K_1\circ \textbf{J}_e$ and $K_2 \circ \textbf{J}_e$ one can take the functions $J_0$ and $\Gamma_0$ given by \eqref{eq:67} and \eqref{eq:68} as the integrals of motion for the Hamiltonian $H\circ \textbf{J}_e$, where $H$ can be arbitrary function $H(\vec{J}, \vec{\Gamma})$ of the vector variables $(\vec{J}, \vec{\Gamma })\in \mathbb{R}^3 \times \dot{\mathbb{R}}^3$.  As a particular case of such Hamiltonian one can  take the one given by \eqref{eq:65}. The fourth integral of motion $K\circ \textbf{J}_e$ for \eqref{eq:65} depends on the concrete choice of the parameters $I_1, I_2, I_3, \lambda_1, \lambda_2, \lambda_3 \in \mathbb{R}$ and the potential function $U\circ\vec{\Gamma }$. The list of the integrals $K$ for many integrable cases one can find in \cite{BF}. Here however, we mention only few of them:
\begin{itemize}
\item[(a)] for Kovalevskaya case distinguished by conditions $I_1=I_2 =2I_3$, $\lambda =0$ and $U(\vec{\Gamma}) = \chi_1 \Gamma_1 + \chi_2 \Gamma_2$ the fourth integral of motion is given by 
\begin{multline}\label{calka1}
K\circ \textbf{J}_e= \frac{1}{4I^2} \left[|a_1|^2|a_2|^2 + |b_1|^2|b_2|^2 - \bar a_1 a_2 \bar b_1 b_2 - a_1 \bar a_2 b_1 \bar b_2\right]^2 +\\
(\chi_1^2 + \chi_2^2 ) \left[ |a_1|^2 |b_2|^2 + |a_2|^2 |b_1|^2 + |a_1|^2 |a_2|^2 + |b_1|^2|b_2|^2 - a_1a_2b_1b_2 - \bar a_1 \bar a_2\bar b_1\bar b_2 + \right.\\
\left.a_1\bar a_2\bar b_1 b_2 + \bar a_1 a_2 b_1 \bar b_2 +i (a_1b_1 - \bar a_1 \bar b_1)(|a_2|^2 + |b_2|^2) + i (a_2 b_2 - \bar a_2 \bar b_2)(|a_1|^2 + |b_1|^2)\right] +\\
\frac{\chi_2}{2I}\left[(a_1b_2 - \bar a_2 \bar b_1 - i(\bar b_1 b_2 + a_1 \bar a_2))(\bar a_1^2 a_2^2 + b_1^2 \bar b_2^2 - \bar a_1 a_2 b_1 \bar b_2) +\right.\\
 \left.(\bar a_1 \bar b_2 - a_2 b_1 +i(b_1 \bar b_2 + \bar a_1 a_2))(a_1^2 \bar a_2^2 + \bar b_1^2 b_2^2 - a_1 \bar a_2 \bar b_1 b_2 )\right] - \\
\frac{\chi_1}{2I} \left[ (b_1 \bar b_2 + \bar a_1 a_2 + i(a_2 b_1 - \bar a_1 \bar b_2))(a_1^2\bar a_2^2 + \bar b_1^2 b_2^2 - a_1 \bar a_2 \bar b_1 b_2) +\right.\\
 \left.(\bar b_1 b_2 + a_1 \bar a_2 + i(a_1b_2-\bar a_2 \bar b_1))(\bar a_1^2 a_2^2 + b_1^2 \bar b_2^2 - \bar a_2 a_2 b_1 \bar b_2)\right] ,
\end{multline}
\item[(b)] for Zhukovski case distinguished by condition $U(\vec{\Gamma}) = 0$ the fourth integral of motion is given by 
\begin{multline}\label{calka2}
K \circ \textbf{J}_e= \vec{J}^2 =  \frac{1}{4}\left(|a_1|^2 -|a_2|^2 -|b_1|^2 + |b_2|^2\right)^2 + |a_1|^2 |a_2|^2 +  |b_1|^2|b_2|^2 -  \bar a_1 a_2 \bar b_1 b_2 -  a_1 \bar a_2b_1 \bar b_2,
\end{multline}
\item[(c)] for Clebsh case distinguished by conditions $\vec{\lambda} =0$ and $U(\vec{\Gamma }) = \frac{\epsilon}{2} (I_1 \Gamma_1^2 + I_2 \Gamma_2^2 + I_3 \Gamma_3^2)$ the fourth integral of motion is given by 
\begin{multline}\label{calka3}
K \circ \textbf{J}_e = \frac{1}{8}\left[\left(|a_1|^2 -|a_2|^2 -|b_1|^2 + |b_2|^2\right)^2 + 4|a_1|^2 |a_2|^2 + 4|b_1|^2|b_2|^2 -  4\bar a_1 a_2 \bar b_1 b_2 -  4a_1 \bar a_2b_1 \bar b_2 -\right. \\
\epsilon\left(I_2I_3(b_1 \bar b_2 +\bar b_1 b_2 + a_1 \bar a_2 + \bar a_1 a_2 + i (a_1b_2 + a_2 b_1 - \bar a_1 \bar b_2 - \bar a_2 \bar b_1))^2\right. +\\
 I_3I_1 (a_1 b_2 +\bar a_1 \bar b_2 -a_2 b_1 - \bar a_2 \bar b_1 + i (b_1 \bar b_2 - \bar b_1 b_2 + \bar a_1 a_2 - a_1 \bar a_2))^2 +\\
\left.I_1I_2 (|a_1|^2 - |a_2|^2 + |b_1|^2 - |b_2|^2 + i(a_1b_1 - a_2b_2 - \bar a_1 \bar b_1 + \bar a_2 \bar b_2))^2\right)\Big]. 
\end{multline}
\end{itemize}

If one interprets $(a^T, b^T) = (a_1, a_2, b_1, b_2)$ as the complex amplitudes of four waves which interact through a non-linear medium, then it is natural to treat \eqref{eq:65} as a Hamiltonian describing this interaction, i.e. as a four-wave mixing Hamiltonian. Let us mention that the non-linear four-wave mixing processes arrise in optics, mechanics, solid body physics  and information theory, see \cite{IL,OW4, Mil}.  Usually, the considered models of four waves systems are solved by numerical or approximative methods.  So, the ones which are strictly integrable are  especially appreciable from physical point of view. In this paper we will not concern the applications of the obtained integrable four-wave mixing models. However, we plan to continue our interests to this subject in subsequent paper.

The Hamiltonian \eqref{eq:65} written in real canonical coordinates $(q,\pi)\in \mathbb{R}^4 \times \mathbb{R}^4$ defined in \eqref{rcn} takes the following form 
\begin{multline}\label{eq:613}
H_{\lambda}\circ \textbf{J}_{e} = \frac{1}{8}\pi^TG \pi + \frac{\lambda_1}{2I_1}(q_0\pi_3+q_1\pi_2-q_2\pi_1 - q_3\pi_0 + \lambda_1) + \\
\frac{\lambda_2}{2I_2}(q_0\pi_1- q_1\pi_0+q_2\pi_3-q_3\pi_4+\lambda_2)+ \\
\frac{\lambda_3}{2I_3}(q_0\pi_2-q_1\pi_3-q_2\pi_0+q_3\pi_1 + \lambda_3) +  (U\circ \vec{\Gamma})(q), 
\end{multline}
where 
\begin{equation}
G=\left(\begin{array}{cccc}
\frac{q_3^2}{I_1} + \frac{q_1^2}{I_2}+\frac{q_2^2}{I_3} & \frac{q_2q_3}{I_1} - \frac{q_0q_1}{I_2}-\frac{q_2q_3}{I_3} & -\frac{q_1q_3}{I_1} + \frac{q_1q_3}{I_2}-\frac{q_0q_2}{I_3} & -\frac{q_0q_3}{I_1} - \frac{q_1q_2}{I_2}+\frac{q_1q_2}{I_3}\\
\frac{q_2q_3}{I_1} - \frac{q_0q_1}{I_2}-\frac{q_2q_3}{I_3} & \frac{q_2^2}{I_1} + \frac{q_0^2}{I_2}+\frac{q_3^2}{I_3} & -\frac{q_1q_2}{I_1} - \frac{q_0q_3}{I_2}+\frac{q_0q_3}{I_3}& -\frac{q_0q_2}{I_1} + \frac{q_0q_2}{I_2}-\frac{q_1q_3}{I_3}\\
-\frac{q_1q_3}{I_1} + \frac{q_1q_3}{I_2}-\frac{q_0q_2}{I_3} & -\frac{q_1q_2}{I_1} - \frac{q_0q_3}{I_2}+\frac{q_0q_3}{I_3} & \frac{q_1^2}{I_1} + \frac{q_3^2}{I_2}+\frac{q_0^2}{I_3} & \frac{q_0q_1}{I_1} - \frac{q_2q_3}{I_2}-\frac{q_0q_1}{I_3}\\
-\frac{q_0q_3}{I_1} - \frac{q_1q_2}{I_2}+\frac{q_1q_2}{I_3}&  -\frac{q_0q_2}{I_1} + \frac{q_0q_2}{I_2}-\frac{q_1q_3}{I_3}& \frac{q_0q_1}{I_1} - \frac{q_2q_3}{I_2}-\frac{q_0q_1}{I_3} & \frac{q_0^2}{I_1} + \frac{q_2^2}{I_2}+\frac{q_1^2}{I_3}
\end{array}\right).
\end{equation}
Let us mention that the real canonical coordinates $(q, \pi)$ and the complex canonical coordinates $(a,b)\in \mathbb{C}^2\times \mathbb{C}^2$ are related by 
\begin{align}
\left(\begin{array}{c}
q_0\\
q_1
\end{array}\right) = \frac{1}{2}(b +\bar b + i (a-\bar a)), \quad & \left(\begin{array}{c}
q_2\\
q_3
\end{array}\right) = \frac{1}{2}(a+ \bar a  + i (b-\bar b)), \\
\left(\begin{array}{c}
\pi_0\\
\pi_1
\end{array}\right) = \frac{1}{2}(a+ \bar a  - i (b-\bar b)), \quad & \left(\begin{array}{c}
\pi_2\\
\pi_3
\end{array}\right) = \frac{1}{2}(b+ \bar b  - i (a-\bar a)).
\end{align}

Assuming in \eqref{eq:613} $\vec{\lambda}=0$ one obtains the Hamiltonian on $\mathbb{R}^4\times \mathbb{R}^4\cong T^*\mathbb{R}^4$ with the kinematic energy defined by $G$ and the potential one by $U\circ \vec{\Gamma}$, where $\vec{\Gamma}$ is given in \eqref{ngamma}.

\subsection{Integrable Hamiltonian systems on $(\mathbb{R}^3\times \dot{\mathbb{R}}^3 , \Omega_\mu )$.}

Lifting the Hamiltonian \eqref{Hht} on $\mathbb{R}^3\times \dot{\mathbb{R}}^3 $ by the symplectic realization $\textbf{J}_{e, \mu} : \mathbb{R}^3\times \dot{\mathbb{R}}^3 \to \textbf{e}(3)^*$ given in \eqref{jemi}, we obtain the Hamiltonian 
\begin{equation}\label{eq:616}
H_{\lambda}\circ \textbf{J}_{e, \mu } = \frac{1}{2}\left(\vec{y}\times \vec{p} +\frac{\mu}{||\vec{y}||} \vec{y}+\vec{\lambda}\right)^T I^{-1} \left(\vec{y}\times \vec{p} +\frac{\mu}{||\vec{y}||} \vec{y}+\vec{\lambda}\right) + U(-\vec{y}). 
\end{equation}
on the symplectic manifold $(\mathbb{R}^3\times \dot{\mathbb{R}}^3 , \Omega_\mu )$.

Since $\textbf{J}_{e, \mu} : \mathbb{R}^3\times \dot{\mathbb{R}}^3 \to \mathbb{R}^3\times \dot{\mathbb{R}}^3 \subset\textbf{e}(3)^*$ is a Poisson map, the function
\begin{equation}
K_2\circ \textbf{J}_{e, \mu} (\vec{p}, \vec{y})= \vec{y}^2, 
\end{equation}
is the second integral of motion for the Hamiltonian $H\circ \textbf{J}_{e, \mu}$, where $H\in C^\infty (\textbf{e}(3)^*, \mathbb{R})$ is arbitrary Hamiltonian function $H=H(\vec{J}, \vec{\Gamma })$. Thus,  $K_2\circ \textbf{J}_{e, \mu} (\vec{p}, \vec{y})$ is integral of motion also for the Hamiltonian system defined by $H_{\lambda}\circ \textbf{J}_{e,\mu}$.

The third integrals of motion for the integrable cases mentioned in the previous subsection take the form:
\begin{itemize}
\item[(a)] for Kovalevskaya case one has
\begin{multline}
K\circ \textbf{J}_{e, \mu }= \left(\frac{1}{2I}\left((y_2p_3-y_3p_2)^2 - (y_2p_1-y_1p_3)^2 + \frac{\mu^2}{\vec{y}^2}(y_1^2-y_2^2) +\right. \right.\\
\left.\left.\frac{2\mu}{||\vec{y}||} (2y_1y_2p_3 -y_1y_3p_2-y_2y_3p_1)\right) +\chi_1y_1 - \chi_2y_2\right)^2 + \\
\left(\frac{1}{I}\left(u_2p_2-y_3p_2+\frac{\mu}{||\vec{y}||}y_1\right)\left(y_3p_1-y_1p_3 +\frac{\mu}{||\vec{y}||}y_2\right)+\chi_1 y_2 +\chi_2 y_1\right)^2,
\end{multline}
\item[(b)] for Zhukovski case one has
\begin{equation}
K\circ \textbf{J}_{e, \mu } = (\vec{y}\times \vec{p})^2 + \mu^2,
\end{equation}
\item[(c)] for Clebsh case one has
\begin{equation}
K\circ \textbf{K}_{e, \mu} = \frac{1}{2}((\vec{y}\times \vec{p})^2 + \mu^2) - \frac{\epsilon}{2}(I_2I_3y_1^2 +I_3I_1 y_2^2 +I_1I_2y_3^2).
\end{equation}  
\end{itemize}

The Hamiltonian \eqref{eq:616} except of the potential $U(-\vec{y})$ has the part which is diagonal quadratic form of the vector 
\begin{equation}
\vec{M} := \vec{y}\times p + \frac{\mu}{||\vec{y}||} \vec{y} + \vec{\lambda}
\end{equation}
which one can interpret as the angular momentum of the system corrected by the gyrostat inner angular momentum $\vec{\lambda}$ and the angular momentum $\frac{\mu}{||\vec{y}||} \vec{y}$ of Dirac monopole. This explains partly the physical sense of the considered Hamiltonian system.

The Hamilton equations given by \eqref{eq:616} are the following
\begin{align}
\nonumber
\frac{d}{dt} \vec{p}  = & (I^{-1}(\vec{y}\times \vec{p}))\times \vec{p} +\frac{\mu}{||\vec{y}||}\left((I^{-1}\vec{y})\times \vec{p} - I^{-1}\left(\vec{y}\times \vec{p} +\frac{\mu}{||\vec{y}||}\vec{y} + \vec{\lambda}\right)\right) + \\
 & \frac{\mu}{\||\vec{y}||^3} \left((\vec{y}\times \vec{p} +\frac{\mu}{||\vec{y}||}\vec{y} + \vec{\lambda})I^{-1}\vec{y}\right)\vec{y}, \\
\frac{d}{dt} \vec{y}  = &  \left(I^{-1}(\vec{y}\times \vec{p} +\frac{\mu}{||\vec{y}||}\vec{y} + \vec{\lambda})\right)\times \vec{y}.
\end{align}

\subsection{Integrable Hamiltonian systems on $(T^*\mathbb{S}^3_{\rho}, \Omega_\nu )$.}

Here we discuss the integrable Hamiltonian system on the symplectic manifold $(T^*\mathbb{S}^3_{\rho}\backslash\left\{\left(\begin{array}{c}
\rho \\
\vec{0}
\end{array}\right)\right\}, \Omega_\nu )$, where $\rho = \sqrt{-2\nu}$, defined by the Hamiltonian $H_{\lambda}\circ \textbf{J}_{e, \nu}$, which written in the Moser coordinates $(\vec{y}, \vec{p})$, defined in \eqref{537} and \eqref{eq:542}, assumes the following form 
\begin{multline}\label{eq:631}
H_{\lambda}\circ \textbf{J}_{e, \nu} = \frac{1}{2} \left(S_{132}\left(\frac{\vec{y}^2 -1}{4}\vec{p} - \frac{1}{2}(\vec{y}\cdot \vec{p}) +\frac{1}{2}\vec{y}\times \vec{p}\right)\right)^TI^{-1} \left(S_{132}\left(\frac{\vec{y}^2 -1}{4}\vec{p} - \frac{1}{2}(\vec{y}\cdot \vec{p}) +\frac{1}{2}\vec{y}\times \vec{p}\right)\right) \\
+(U\circ \vec{\Gamma}) (\vec{y}),
\end{multline}
where $S_{132}$ is the permutation matrix 
\begin{equation}
S_{132} = \left(\begin{array}{ccc}
0 & 0 & 1\\
1 & 0 & 0\\
0 & 1 & 0
\end{array}\right)
\end{equation}
and  $\textbf{J}_{e, \nu }: T^*\mathbb{S}^3_\rho \to \textbf{e}(3)^*$ is given in \eqref{eq:546} and \eqref{eq:547}. Substituting $(q,\pi)\in \mathbb{R}^4\times \mathbb{R}^4 \cong T^*\mathbb{R}^4$ given by \eqref{538i}  and \eqref{541} into \eqref{eq:469} we obtain the second integral of motion 
\begin{equation}
\tilde{J}_0 (\vec{y}, \vec{p}) = -\frac{1}{2} \left(\frac{\vec{y}^2-1}{2}p_2 + (\vec{y}\cdot \vec{p})y_2 - y_3p_1 +y_1p_3\right)
\end{equation}
for the Hamiltonian \eqref{eq:631}, corresponding to the Casimir function $K_1 \in C^\infty(\textbf{e}(3)^*, \mathbb{R})$, see \eqref{casimirs}.

As in two previous subsections, below we present the third integrals of motion for the lifting of Kovalevska, Zhukovski and Clebsh Hamiltonians: 
\begin{itemize}
\item[(a)] for Kovalevska case the third integral of motion is given by 
\begin{multline}
K\circ \textbf{J}_{e, \nu} = \left(\frac{1}{8I}\left[\left(\frac{\vec{y}^2-1}{2}p_3 -(\vec{y}\cdot \vec{p})y_3 +y_1p_2-y_2p_1\right)^2-\left(\frac{\vec{y}^2-1}{2}p_1 -(\vec{y}\cdot \vec{p})y_1 +y_2p_3-y_3p_2\right)^2\right]\right.+ \\
\left.\frac{2\rho^2}{(1+\vec{y}^2)^2} (\chi_2(2y_1y_2 -(\vec{y}^2-1)y_3) -\chi_1 (2y_2y_3 +(\vec{y}^2-1)y_1))\right)^2 +\\
\left(\frac{1}{4I}\left(\frac{\vec{y}^2-1}{2}p_3 -(\vec{y}\cdot \vec{p})y_3 +y_1p_2-y_2p_1\right)\left(\frac{\vec{y}^2-1}{2}p_1 -(\vec{y}\cdot \vec{p})y_1 +y_2p_3-y_3p_2\right) - \right.\\
\left.\frac{2\rho^2}{(1+\vec{y}^2)^2} (\chi_1 (2y_1y_2 -(\vec{y}^2-1)y_3) + \chi_2 (2y_2y_3 +(\vec{y}^2-1)y_1))\right)^2,
\end{multline}
\item[(b)] for Zhukovski case the third integral of motion is given by 
\begin{equation}
K\circ \textbf{J}_{e, \nu} = \frac{1}{4}\left(\frac{(\vec{y}^2-1)^2}{4} \vec{p}^2 +(\vec{y}\cdot \vec{p})^2 + (\vec{y}\times \vec{p})^2\right),
\end{equation}
\item[(c)] for Clebsh case third integral of motion is given by
\begin{multline}
K\circ \textbf{J}_{e, \nu} = \frac{1}{8}\left(\frac{(\vec{y}^2-1)^2}{4} \vec{p}^2 +(\vec{y}\cdot \vec{p})^2 + (\vec{y}\times \vec{p})^2\right)-\\
\frac{2\epsilon \rho^4}{(1+\vec{y}^2)^4}\left(I_2I_3 ((\vec{y}^2-1)y_1 + 2 y_2y_3)^2+I_3I_1(2y_1y_2-(\vec{y}^2-1)y_3)^2+ \right.\\
\left.I_1I_2 (\frac{1}{4}(\vec{y}-1)^2+y_1^2-y_2^2-y_3^2)^2\right).
\end{multline}
\end{itemize}

The intricate form of the Hamiltonian \eqref{eq:631} does not permit to find a reasonable physical interpretation of the dynamics of obtained model. However, such possibility is not excluded.

At the end let us mention that the symplectic manifolds obtained here and the integrable Hamiltonian systems on them arise naturally from the Hamiltonian systems related to the rigid body theory. Thus, we expect that these systems may have many various aplications, e.g. in nonlinear optics and mechanics. 

\appendix
\section{}\label{B}
Below we give a few definitions of notions, see e.g. \cite{We,W}, used in the paper.  
\begin{defn}\label{B1}
Poisson manifolds $P_1$ and $P_2$ form a symplectic dual pair if there exists symplectic manifold $M$ and Poisson maps 
\begin{equation}
\begin{tikzcd}
    & M \arrow{ld}[swap]{J_1} \arrow{rd}{J_2} &     \\
P_1 &                                       & P_2
\end{tikzcd}
\end{equation}
with symplectically orthogonal fibres. The above implies 
\begin{equation}
\{J_1^*f_1, J_2^*f_2\} =0
\end{equation}
for all $f_1 \in C^\infty (P_1, \mathbb{R})$ and $f_2 \in C^\infty (P_2, \mathbb{R})$. 
\end{defn}

\begin{defn}\label{B2}
Poisson map $\phi: P_1 \to P_2$ is complete if for any complete Hamiltonian vector field $X_H$ defined by $H\in C^\infty(P_2, \mathbb{R})$ follows the completeness of the Hamiltonian vector field $X_{H\circ \phi}$.  
\end{defn}

\begin{defn}\label{B3}
A surjective submersive symplectic realization $\Phi: M \to P$  of a Poisson space $P$ is called a full symplectic realization of $P$.
\end{defn}

\section{}\label{A}
Let us consider symplectic manifold $(T^*Q, d\Theta_Q)$, where $Q$ is a smooth manifold and $\Theta_Q$ is Liouville $1$-form on the cotangent bundle $T^*Q$. Suppose that $R$ is a submanifold of $Q$. By $\iota : R\to Q$ we denote the inclusion map. By $\pi^*_Q: T^*Q \to Q$ and $\pi^*_R: T^*R\to R$ we denote the bundle projections on the base. The Liouville form on $T^*R$ we denote by $\Theta_R$. The manifold 
$$ T^*Q|_{R} := \{(r, \varphi)\in R\times T^*Q: \iota (r) = \pi^*_Q(\varphi)\}$$
can be considered as a vector bundle $\pi^*_Q: T^*Q|_{R}\to \iota(R)\cong R$ over $R$. The embedding $\iota:R\to Q$ defines (tangent) morphism $T\iota: TR\to TQ$ of the tangent bundles and thus, by the duality, the submersive and surjective morphism $T^*\iota:T^*Q|_{R}\to T^*R$ of the bundle $T^*Q|_{R}$  on the bundle $T^*R$ cotangent to $R$. One defines $T^*\iota:T^*Q|_{R}\to T^*R$ as follows
\begin{equation}
T^*\iota:T^*Q|_{R}\ni (r, \varphi) \mapsto T^*\iota (r, \varphi) = \varphi \circ T(r) \in  T^*R. 
\end{equation}
We illustrate the above objects and the corresponding morphisms in the following diagram 
\begin{equation}\label{diagap}
\begin{tikzcd}
T^*Q|_R \arrow{r}{T^*\iota} \arrow{d}{\pi^*_Q} & T^*R \arrow{d}{\pi^*_R} \\
\iota (R)                                          & R \arrow{l}[swap]{\iota}    
\end{tikzcd}.
\end{equation}
\begin{prop}\label{prop:ap}
(i) Let $I_R: T^*Q|_R \hookrightarrow T^*Q$ be the inclusion morphism of the vector bundles. Then we have
\begin{equation}\label{a3}
(I_R)^*\Theta_Q = (T^*\iota)^*\Theta_R.
\end{equation} 
(ii) The reduced symplectic manifold $(T^*Q|_R/\sim , d \Theta_Q)$ is isomorphic with $(T^*R, d\Theta_R)$, where $T^*Q|_R/\sim$ is the quotient of $T^*Q|_R$ by the degeneracy leaves of $(I_R)^*d\Theta_Q$ and $d\Theta_Q$ is the reduction of $d\Theta_Q$ to $T^*Q|_R/\sim$. 
\end{prop}
\begin{proof}
(i) For $\varphi \in T^*Q|_R$ and any $\xi(\varphi) \in T_\varphi (T^*Q|_R)$ we have 
\begin{multline}\label{a4}
\langle I_R^*\Theta_Q (\varphi), \xi (\varphi)\rangle = \langle \Theta_Q(I_R(\varphi)), TI_R(\varphi)\xi (\varphi)\rangle =\\
\langle I_R(\varphi), T\pi^*_Q (\varphi) TI_R(\varphi)\xi (\varphi) \rangle = \langle \varphi , T\pi^*_Q (\varphi)\xi(\varphi)\rangle. 
\end{multline}
On the other hand, from the diagram \eqref{diagap}, we have 
\begin{multline}\label{a5}
\langle \left((T^*\iota)^*\Theta_R\right) (\varphi) , \xi (\varphi)\rangle = \langle \Theta_R(T^*\iota (\varphi), T(T^*\iota)(\varphi)\xi (\varphi)\rangle = \\
\langle T^*\iota (\varphi) , T\pi^*_R (T^*\iota (\varphi ))\circ T(T^*\iota)(\varphi)\xi (\varphi)\rangle = \langle T^*\iota (\varphi), T(\pi^*_R \circ T^*\iota )(\varphi) \xi (\varphi )\rangle =\\
\langle \varphi \circ T\iota (r) , T(\pi^*_R \circ T^*\iota )(\varphi) \xi (\varphi )\rangle = \langle \varphi , T\iota (r) \circ T(\pi^*_R \circ T^*\iota )(\varphi) \xi (\varphi )\rangle =\\
\langle \varphi , T\iota ((\pi^*_R\circ T^*\iota)(\varphi)\circ T(\pi^*_R \circ T^*\iota))(\varphi) \xi (\varphi) \rangle = \langle \varphi , T(\iota \circ \pi^*_R \circ T^*\iota )(\varphi) \xi (\varphi) \rangle = \\
\langle \varphi , T\pi^*_Q (\varphi )\xi (\varphi) \rangle .
\end{multline}
Proving the above equalities we used $\pi^*_Q(\varphi ) = \iota (r) = (\iota \circ \pi^*_R\circ T^*\iota )(\varphi )$ and $\iota (r) = r$. 
Comparing \eqref{a4} with \eqref{a5} we obtain \eqref{a3}.\\
(ii) It follows from \eqref{a3} that degeneracy leaves of $(I_R)^*d\Theta_Q$ are the same as the ones for $(T^*\iota)^*\Theta_R$. Thus, one has the isomorphism $(T^*Q|_R/\sim , d \Theta_Q)\cong (T^*R, d\Theta_R)$.
\end{proof}


\thebibliography{99}
\bibitem{BF} Bolsinov A. V., Fomenko A. T., \textit{Integrable Hamiltonian systems}, 
\bibitem{C}  Clebsch A., \textit{Uber die Bewegung eines Korpers in einer Flussigkeit}, Math. Annalen, 3 (1871), 238- 262.
\bibitem{ZU} Dufour J.P., Zung N.T., \textit{Poisson structures and their normal forms}, Birkhäuser; 2005
\bibitem{G1} Guillemin V., Sternberg S., \textit{Remarks on a paper of Hermann}, Trans. Amer. Math. Soc. 130 (1968)110-116
\bibitem{G2} Guillemin V., Sternberg S.,\textit{The moment map and collective motion}, Ann. Physics 127 (1980) 220-253
\bibitem{IL} Imamoglu A., Lukin M. D., \textit{Nonlinear Optics and quantum entanglement of ultraslow single photons}, Phys. Rev. Lett. 84(7), 1419 (2000).

\bibitem{I}  Iwai T., \textit{A dynamical group $SU(2, 2)$ and its use in the MIC-Kepler problem}, J. Phys. A: Math. Gen. 26
(1993), 609–630.
\bibitem{SK} Kowalewski S., \textit{Sur le probleme de la rotation d'un corps solide autour d'un point fixe}, Acta Math., 12 (1889), 177-232 
\bibitem{L} Lipkin H.J., Meshkov N., Glick A.J., \textit{Validity of Many-Body Approximation Method: Exact Solutions and Perturbation Theory}, Nuclear Physics 62 (1965
\bibitem{Mos} Moser J.K., \textit{Regularization of Kepler's problem and the averaging method on a manifold}, Commun. Pure Appl. Math.23, 609–636 (1970)
\bibitem{O} Odzijewicz A., \textit{A conformal holomorphic field theory}, Commun. Math. Phys. 107 (1986), no. 4, 561-575.
\bibitem{O2} Odzijewicz A., \textit{Coherent states and geometric quantization}, Commun. Math. Phys. 150 (1992), no. 2, 385-413.
\bibitem{OSW} Odzijewicz A., Sliżewska A., Wawreniuk E., \textit{A Family of Integrable Perturbed Kepler Systems}, Russ. J. Math. Phys., Vol. 26, No. 3 (2019) 
\bibitem{OS} Odzijewicz A., Świetochowski M., \textit{Coherent states map for MIC–Kepler system}, J. Math. Phys. 38, 5010 (1997)
\bibitem{OW4} Odzijewicz A., Wawreniuk E., \textit{An integrable (classical and quantum) four-wave mixing Hamiltonian system}, J. Math. Phys. 61 (2020), no. 7, 1-18, 
\bibitem{P} Penrose R., \textit{Twistor algebra}, J. Math. Phys 8, 345-366 (1967)
\bibitem{Pe} Perelomov A.M., \textit{Integrable Systems of Classical Mechanics and Lie Algebras}, Birkhäuser Basel (1990)
\bibitem{RM} Marsden J.E., Ratiu T.S., \textit{Introduction to Mechanics and symmetry}, Springer, 1994
\bibitem{Mil} Milburn W., Walls D.F., \textit{Quantum Optics}, 1st ed., Springer-Verlag, 1994
\bibitem{MF} Mishchenko A. S., Fomenko A. T., \textit{Generalized Liouville method of integration of Hamiltonian systems}, Funct. Anal. Appl. 12 (1978) 113-121 (1978)
\bibitem{We} Weinstein A., \textit{The local structure of Poisson manifolds}, J. Differential Geom. 18(3): 523-557 (1983)
\bibitem{W} Weinstein A., Cannas da Silva A., \textit{Geometric Models for Noncommutative Algebras}, American Mathematical Society
\bibitem{Z} Zhukovski N.E., \textit{On the motion of a rigid body having cavities filled with homogeneous liquid}, Zh. Russk. Fiz-Khim. Obsch., 17 (1885), No.6, P. 81{113; No.7, P. 145{149; No.8, P. 231{280.
\end{document}
\begin{multline}\label{eq:631}
H_{\lambda}\circ \textbf{J}_{e, \nu}(\vec{y}, \vec{p}) = \frac{1}{8} \left[ p_1^2 \left(\frac{(y_2+y_1y_3)^2}{I_1} + \frac{\left(\frac{\vec{y}^2 -1}{2}- y_1^2\right)^2}{I_2} + \frac{(y_3-y_1y_2)^2}{I_3}\right) + \right.\\
p_2^2 \left(\frac{(y_1-y_2y_3)^2}{I_1} + \frac{(y_3+y_1y_2)^2}{I_2} + \frac{\left(\frac{\vec{y}^2-1}{2}- y_2^2\right)^2}{I_3}\right) + 
p_3^2\left( \frac{\left(\frac{\vec{y}^2-1}{2}- p_3^2\right)^2}{I_1} + \frac{(y_2-y_1y_3)^2}{I_2} + \frac{(y_1 + y_2y_3)^2}{I_3}\right) + \\
2p_1p_2 \left(\frac{1}{I_1}(y_1y_2(y_3^2-1)+y_3(y_2^2-y_1^2))+ \frac{1}{I_2}(y_1y_2+y_3)\left(y_1^2-\frac{\vec{y}^2-1}{2}\right) + \frac{1}{I_3}(y_1y_2-y_3)\left(y_2^2-\frac{\vec{y}^2-1}{2}\right)\right) + \\
2p_1p_3 \left(\frac{1}{I_1}(y_1y_3+y_2)\left(y_3^2-\frac{\vec{y}^2-1}{2}\right)+\frac{1}{I_2}(y_1y_3-y_2)\left(y_1^2-\frac{\vec{y}^2-1}{2}\right)+\frac{1}{I_3}\left(y_1y_3\left(y_2^2-1\right)+ y_2(y_1^2-y_3^2)\right)\right)+ \\
\left.2p_2p_3 \left(\frac{1}{I_1}(y_2y_3-y_1)\left(y_3^2- \frac{\vec{y}^2-1}{2}\right) + \frac{1}{I_2}(y_2y_3(y_1^2-1) + y_1(y_3^2 - y_1^2)) + \frac{1}{I_3}(y_2y_3+y_1)\left(y_2^2- \frac{\vec{y}^2-1}{2}\right)\right)\right]\\
+ \frac{\lambda_1}{2I_1}\left(\frac{\vec{y}^2-1}{2}p_3 - (\vec{y}\cdot \vec{p})y_3 + y_1p_2 - y_2p_1 + \lambda_1\right) + 
\frac{\lambda_2}{2I_2}\left(\frac{\vec{y}^2-1}{2}p_1 - (\vec{y}\cdot \vec{p})y_1 + y_2p_3 - y_3p_2 + \lambda_2\right)+\\
\frac{\lambda_1}{2I_1}\left(\frac{\vec{y}^2-1}{2}p_2 - (\vec{y}\cdot \vec{p})y_2 + y_3p_1 - y_1p_3 + \lambda_3\right) + (U\circ \vec{\Gamma})(\vec{y}),
\end{multline}
H_{\lambda}\circ \textbf{J}_{e, \mu }(\vec{p}, \vec{y}) = \frac{1}{2} \vec{p}^T G \vec{p} +\frac{\mu}{||\vec{y}||} \vec{p}\cdot \left((I^{-1}\vec{y})\times \vec{y}\right) +\frac{1}{2}\frac{\mu^2}{\vec{y}^2} \vec{y}^TI^{-1}\vec{y} + \\
(I^{-1}\vec{\lambda})\cdot \left(\vec{y}\times \vec{p} + \frac{\mu}{||\vec{y}||}\vec{y}\right) + \vec{\lambda} I^{-1}\vec{\lambda} + U(-\vec{y}),
where
\begin{equation}
G= \left(\begin{array}{ccc}
\frac{y_2^2}{I_3}+ \frac{y_3^2}{I_2} & \frac{-y_1y_2}{I_3} & \frac{-y_1y_3}{I_2}\\
\frac{-y_1y_2}{I_3} & \frac{y_1^2}{I_3}+ \frac{y_3^2}{I_1} & \frac{-y_2y_3}{I_1}\\
\frac{-y_1y_3}{I_2} & \frac{-y_2y_3}{I_1} & \frac{y_1^2}{I_2}+ \frac{y_2^2}{I_1},
\end{array}\right)
\end{equation}

\textbf{1. Hamiltonian system on $(\widetilde{\mathbb{T}}\cong \mathbb{R}^4\times \mathbb{R}^4, \Omega)$}

Using \eqref{eq:je} and .... one can realize the general Hamiltonian \eqref{Hht} describing the dynamics of a gyrostat with a fixed point as a function on $\mathbb{R}^4\times \mathbb{R}^4$ given by 
\begin{multline}
\textbf{J}_e^*H_{\lambda} = \frac{1}{8}\pi^TG \pi + \frac{\lambda_1}{2I_1}(q_0\pi_3+q_1\pi_2-q_2\pi_1 - q_3\pi_0 + \lambda_1) + \\
\frac{\lambda_2}{2I_2}(q_0\pi_1- q_1\pi_0+q_2\pi_3-q_3\pi_4+\lambda_2)+ \\
\frac{\lambda_3}{2I_3}(q_0\pi_2-q_1\pi_3-q_2\pi_0+q_3\pi_1 + \lambda_3) +  U(q), 
\end{multline}
where 
\begin{equation}
G=\left(\begin{array}{cccc}
\frac{q_3^2}{I_1} + \frac{q_1^2}{I_2}+\frac{q_2^2}{I_3} & \frac{q_2q_3}{I_1} - \frac{q_0q_1}{I_2}-\frac{q_2q_3}{I_3} & -\frac{q_1q_3}{I_1} + \frac{q_1q_3}{I_2}-\frac{q_0q_2}{I_3} & -\frac{q_0q_3}{I_1} - \frac{q_1q_2}{I_2}+\frac{q_1q_2}{I_3}\\
\frac{q_2q_3}{I_1} - \frac{q_0q_1}{I_2}-\frac{q_2q_3}{I_3} & \frac{q_2^2}{I_1} + \frac{q_0^2}{I_2}+\frac{q_3^2}{I_3} & -\frac{q_1q_2}{I_1} - \frac{q_0q_3}{I_2}+\frac{q_0q_3}{I_3}& -\frac{q_0q_2}{I_1} + \frac{q_0q_2}{I_2}-\frac{q_1q_3}{I_3}\\
-\frac{q_1q_3}{I_1} + \frac{q_1q_3}{I_2}-\frac{q_0q_2}{I_3} & -\frac{q_1q_2}{I_1} - \frac{q_0q_3}{I_2}+\frac{q_0q_3}{I_3} & \frac{q_1^2}{I_1} + \frac{q_3^2}{I_2}+\frac{q_0^2}{I_3} & \frac{q_0q_1}{I_1} - \frac{q_2q_3}{I_2}-\frac{q_0q_1}{I_3}\\
-\frac{q_0q_3}{I_1} - \frac{q_1q_2}{I_2}+\frac{q_1q_2}{I_3}&  -\frac{q_0q_2}{I_1} + \frac{q_0q_2}{I_2}-\frac{q_1q_3}{I_3}& \frac{q_0q_1}{I_1} - \frac{q_2q_3}{I_2}-\frac{q_0q_1}{I_3} & \frac{q_0^2}{I_1} + \frac{q_2^2}{I_2}+\frac{q_1^2}{I_3}
\end{array}\right)
\end{equation}

$\bullet $ Kovalevskaya top (do poprawy):
\begin{equation}
H = \frac{1}{2I} p^TGp + 2\chi(q_1q_2+q_3q_4) + 2\chi (q_2q_3-q_1q_4), 
\end{equation}
where 
\begin{equation}
G = \left(\begin{array}{cccc}
q_2^2+2q_3^2+q_4^2 & -(q_1q_2+q_3q_4) & -2q_1q_3 & q_2q_3-q_1q_4 \\
-(q_1q_2+q_3q_4) & q_1^2 +q_3^2 +2q_4^2 & q_1q_4 - q_2q_3 & -2q_2q_4 \\
-2q_1q_3 & q_1q_4 - q_2q_3 & 2q_1^2 + q_2^2 + q_4^2 & -(q_1q_2+q_3q_4)\\
q_2q_3-q_1q_4  & -2q_2q_4 & -(q_1q_2+q_3q_4) & q_1^2+2q_2^2 +q_3^2
\end{array}\right).

\textbf{Hamiltonians on $(\mathbb{R}^3\times \dot{\mathbb{R}}^3 , \Omega_\mu )$}

$\bullet $ General case:
\begin{multline}
H = \frac{1}{2} \vec{p}^T G \vec{p} +\frac{\mu}{||\vec{y}||} \vec{p}\cdot \left((I^{-1}\vec{y})\times \vec{y}\right) +\frac{1}{2}\frac{\mu^2}{\vec{y}^2} \vec{y}^TI^{-1}\vec{y} + \\
(I^{-1}\vec{\lambda})\cdot \left(\vec{y}\times \vec{p} + \frac{\mu}{||\vec{y}||}\vec{y}\right) + \vec{\lambda} I^{-1}\vec{\lambda} + U(\vec{y}). 
\end{multline}
where
\begin{equation}
G= \left(\begin{array}{ccc}
\frac{y_2^2}{I_3}+ \frac{y_3^2}{I_2} & \frac{-y_1y_2}{I_3} & \frac{-y_1y_3}{I_2}\\
\frac{-y_1y_2}{I_3} & \frac{y_1^2}{I_3}+ \frac{y_3^2}{I_2} & \frac{-y_2y_3}{I_1}\\
\frac{-y_1y_3}{I_2} & \frac{-y_2y_3}{I_1} & \frac{y_1^2}{I_2}+ \frac{y_2^2}{I_1}
\end{array}\right)
\end{equation}

$\bullet $ Kovalevskaya top:
\begin{multline}
H_K = \frac{1}{2I} \vec{p}^T \left(\begin{array}{ccc}
2y_2^2+y_3^2 & -2y_1y_2 & -y_1y_3 \\
-2y_1y_2 & y_3^2+2y_1^2 & -y_2y_3 \\
-y_1y_3 & -y_2y_3 & y_1^2+y_2^2 
\end{array}\right) \vec{p} + \\
\frac{\mu}{I||\vec{y}||}(p_2y_1y_3 -p_1y_2y_3) + \frac{\mu^2}{2I||\vec{y}||}\left(1+ \frac{y_3^2}{||\vec{y}||}\right) - \chi_1 y_1 - \chi_2 y_2. 
\end{multline}

\end{equation}
Rescaling $\zeta \mapsto \frac{1}{\sqrt{-\nu}}\zeta$ by the radius of $\mathbb{S}^3_{\sqrt{-\nu}}$ we can consider $\zeta$ as an element $\zeta \in \mathbb{S}^3$ of the $3$-sphere of radius one and thus, after rescalling, we find that $(\zeta , P)\in \mathbb{S}^3\times iH_0(2)\cong T^*\mathbb{S}^3$ is an element of the cotangent bundle $T^*\mathbb{S}^3 =T^*SU(2)$ of the special unitary group $SU(2)\cong \mathbb{S}^3$. Then the symplectic form $\Psi^*\Omega_\nu$ one can rewrite as 
\begin{equation}
\Psi^*\Omega_\nu = \nu d\mbox{Tr}(P(\zeta d\zeta^+ - d\zeta \zeta^+)), 
\end{equation}
where $(\zeta , P)\in SU(2)\times iH_0(2)$.  Let us note that the differential $1$-form $\zeta d\zeta^+ - d\zeta \zeta^+$ assumes values in $iH(2)$ and isomorphism between the Lie algebra $\textbf{su}(2)= iH_0(2)$ and its dual $\textbf{su}(2)^*$ is given by the trace $\mbox{Tr}$.